\documentclass[11pt]{article}

\usepackage[preprint]{acl}

\usepackage{times}
\usepackage{latexsym}

\usepackage[T1]{fontenc}

\usepackage[utf8]{inputenc}

\usepackage{microtype}

\usepackage{inconsolata}

\usepackage{graphicx}

\usepackage{bm}
\usepackage{geometry}
\usepackage{multirow}
\usepackage{booktabs}
\usepackage{algorithm}
\usepackage{algorithmicx}
\usepackage{algpseudocode}
\usepackage{amsmath}
\usepackage{comment}
\usepackage{ulem}
\usepackage{tikz}
\usepackage{pgfplots}
\pgfplotsset{compat=1.18}
\usepackage{amsthm}
\usepackage{amssymb}
\usepackage{enumitem}
\usepackage{subcaption}
\usepackage[most]{tcolorbox}

\allowdisplaybreaks[4]

\makeatletter
\def\th@acmplain{%
  \setlength{\parindent}{0pt}%
  \thm@headfont{\bfseries\scshape}
  \thm@notefont{\normalfont\scshape}
}
\makeatother

\definecolor{lightgraybg}{RGB}{235,235,235}
\newtcolorbox{graybox}{
    colback=lightgraybg,
    colframe=lightgraybg,
    fonttitle=\bfseries,
    boxsep=2pt,
    left=0pt,
    right=0pt,
    top=0pt,
    bottom=0pt,
    breakable,
    arc=0mm,
    before skip=\topsep,
    after skip=\topsep
}

\newtheorem{theorem}{Theorem}[section]

\newtheorem{corollary}[theorem]{Corollary}
\newtheorem{lemma}[theorem]{Lemma}
\newtheorem{proposition}[theorem]{Proposition}
\newtheorem{definition}[theorem]{Definition}

\newtheorem{fact}[theorem]{Fact}

\tcolorboxenvironment{theorem}{%
    colback=lightgraybg,
    colframe=lightgraybg,
    fonttitle=\bfseries,
    boxsep=2pt,
    left=0pt,
    right=0pt,
    top=0pt,
    bottom=0pt,
    breakable,
    arc=0mm,
    before skip=\topsep,
    after skip=\topsep
}

\tcolorboxenvironment{problem}{%
    colback=lightgraybg,
    colframe=lightgraybg,
    fonttitle=\bfseries,
    boxsep=2pt,
    left=0pt,
    right=0pt,
    top=0pt,
    bottom=0pt,
    breakable,
    arc=0mm,
    before skip=\topsep,
    after skip=\topsep
}

\tcolorboxenvironment{claim}{%
    colback=lightgraybg,
    colframe=lightgraybg,
    fonttitle=\bfseries,
    boxsep=2pt,
    left=0pt,
    right=0pt,
    top=0pt,
    bottom=0pt,
    breakable,
    arc=0mm,
    before skip=\topsep,
    after skip=\topsep
}

\tcolorboxenvironment{corollary}{%
    colback=lightgraybg,
    colframe=lightgraybg,
    fonttitle=\bfseries,
    boxsep=2pt,
    left=0pt,
    right=0pt,
    top=0pt,
    bottom=0pt,
    breakable,
    arc=0mm,
    before skip=\topsep,
    after skip=\topsep
}

\tcolorboxenvironment{lemma}{%
    colback=lightgraybg,
    colframe=lightgraybg,
    fonttitle=\bfseries,
    boxsep=2pt,
    left=0pt,
    right=0pt,
    top=0pt,
    bottom=0pt,
    breakable,
    arc=0mm,
    before skip=\topsep,
    after skip=\topsep
}

\tcolorboxenvironment{proposition}{%
    colback=lightgraybg,
    colframe=lightgraybg,
    fonttitle=\bfseries,
    boxsep=2pt,
    left=0pt,
    right=0pt,
    top=0pt,
    bottom=0pt,
    breakable,
    arc=0mm,
    before skip=\topsep,
    after skip=\topsep
}

\tcolorboxenvironment{definition}{%
    colback=lightgraybg,
    colframe=lightgraybg,
    fonttitle=\bfseries,
    boxsep=2pt,
    left=0pt,
    right=0pt,
    top=0pt,
    bottom=0pt,
    breakable,
    arc=0mm,
    before skip=\topsep,
    after skip=\topsep
}

\tcolorboxenvironment{remark}{%
    colback=lightgraybg,
    colframe=lightgraybg,
    fonttitle=\bfseries,
    boxsep=2pt,
    left=0pt,
    right=0pt,
    top=0pt,
    bottom=0pt,
    breakable,
    arc=0mm,
    before skip=\topsep,
    after skip=\topsep
}

\tcolorboxenvironment{fact}{%
    colback=lightgraybg,
    colframe=lightgraybg,
    fonttitle=\bfseries,
    boxsep=2pt,
    left=0pt,
    right=0pt,
    top=0pt,
    bottom=0pt,
    breakable,
    arc=0mm,
    before skip=\topsep,
    after skip=\topsep
}

\tcolorboxenvironment{question}{%
    colback=lightgraybg,
    colframe=lightgraybg,
    fonttitle=\bfseries,
    boxsep=2pt,
    left=0pt,
    right=0pt,
    top=0pt,
    bottom=0pt,
    breakable,
    arc=0mm,
    before skip=\topsep,
    after skip=\topsep
}

\newcommand{\R}{\mathbb{R}}

\newcommand{\wh}[1]{\widehat{#1}}

\newcommand{\ov}[1]{\overline{#1}}

\newcommand{\argmin}{\mathop{\arg\min}}

\newcommand{\thetab}{\boldsymbol{\theta}}

\newcommand{\hb}{\mathbf{h}}

\newcommand{\ub}{\mathbf{u}}
\newcommand{\vb}{\mathbf{v}}

\newcommand{\xb}{\mathbf{x}}

\newcommand{\zb}{\mathbf{z}}

\newcommand{\Ab}{\mathbf{A}}
\newcommand{\Bb}{\mathbf{B}}
\newcommand{\Cb}{\mathbf{C}}

\newcommand{\Ib}{\mathbf{I}}

\newcommand{\Lb}{\mathbf{L}}

\newcommand{\Pb}{\mathbf{P}}
\newcommand{\Qb}{\mathbf{Q}}
\newcommand{\Rb}{\mathbf{R}}

\newcommand{\Ub}{\mathbf{U}}
\newcommand{\Vb}{\mathbf{V}}
\newcommand{\Wb}{\mathbf{W}}
\newcommand{\Xb}{\mathbf{X}}

\newcommand{\Dcal}{\mathcal{D}}

\newcommand{\Qcal}{\mathcal{Q}}
\newcommand{\Rcal}{\mathcal{R}}

\newcommand{\diag}{\mathrm{diag}}
\newcommand{\tr}{\mathrm{tr}}

\input{__append_toc}

\title{%
\makebox[0pt][r]{%
    \raisebox{-0.9em}{%
        \includegraphics[
            height=4em,
            trim=8 8 8 8,
            clip
        ]{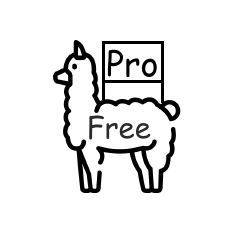}%
    }%
    \hspace{-0.75em}%
}%
Your ``Pro'' LLM Subscription May Actually Be ``Free'':\\[-0.35em]
Exposing Fingerprint Spoofing Risks in LLM Inference Services%
}

\renewcommand{\thefootnote}{\arabic{footnote}}
\author{Jiahao Zhang\textsuperscript{*},
  Xiuyu Li\textsuperscript{*}, Suhang Wang \\
  The Pennsylvania State University, USA \\
  \texttt{jiahao.zhang@psu.edu, disangan233@gmail.com, szw494@psu.edu} 
}

\begin{document}
\maketitle

\begingroup
\renewcommand{\thefootnote}{\fnsymbol{footnote}}
\footnotetext[1]{Equal contribution.}
\endgroup

\begin{abstract}
As Large Language Model (LLM) APIs become ubiquitous, users increasingly rely on black-box fingerprinting to verify that providers are serving the advertised premium models. However, these methods may overlook adversarial providers who manipulate model weights to cheat the fingerprint process. We introduce a novel threat termed \textit{fingerprint spoofing}, where a malicious provider stealthily serves a weaker model that has been parameter-efficiently fine-tuned to mimic a stronger model, thereby evading user-side fingerprinting.
We first formally prove that user-side resource constraints (i.e., finite query budgets and weak fingerprinting classifiers) make current fingerprinting vulnerable to fingerprint spoofing. Guided by this theoretical analysis, we propose \textbf{GhostPrint}, a cost-effective attack framework leveraging surrogate modeling, reward-ranked fine-tuning, and knowledge distillation. Extensive evaluations in both static and continual fingerprinting settings demonstrate that GhostPrint allows weak models to consistently bypass representative fingerprint methods while maintaining utility at a low fine-tuning cost, exposing a critical vulnerability in current LLM fingerprinting pipelines.
\end{abstract}

\section{Introduction}

LLMs have shown great ability on various tasks. A growing number of third-party providers offer API-based inference services for LLMs, such as OpenRouter\footnote{\url{https://openrouter.ai/}} and other Model-as-a-Service providers~\cite{liang2023holistic,gao2025model,chen2024frugalgpt}. These platforms serve a dual role: providing foundational architecture for downstream applications (e.g., coding agents~\cite{yao2023react,shinn2023reflexion,ning2026code}, retrieval-augmented generation~\cite{fan2025towards,liu2025exposing,yang2026query,nguyen2026urag}, personalized chat-bots~\cite{jiang2025know,jiang2025personamem,huang2026towards}, multi-modal reasoning~\cite{zhang2025sua,wu2026image,fang2026agent}) and acting as essential infrastructure for academic evaluation~\cite{chiang2024chatbot,qin2024toolllm}.

Given the foundational role of these API services, ensuring the integrity of the served models has become a pressing necessity~\cite{cai2025you,velasco2025your,zhang2026shadowapis,zhang2026resource}. From a user's perspective, a practical approach to audit whether LLM APIs are serving the advertised model relies on LLM fingerprinting (i.e., model auditing) methods~\cite{gao2025model,gubri2024trap,sun2025idiosyncrasies,pasquini2025llmmap}. The key aim of these methods is to probe a black-box API with a small number of queries and efficiently analyze the responses to verify the model's identity. This relies on the common assumption that distinct response styles to specific prompts are sufficient to classify models~\cite{pasquini2025llmmap,sun2025idiosyncrasies}, making fingerprinting a trustworthy tool.

However, fingerprinting inherently assumes that the models remain static and are not adversarially manipulated by API providers, which might not hold in practice. For instance, an inference provider serving open-source models might adversarially fine-tune a weak model to closely mimic the fingerprint of a stronger one, bypassing user audits, for profit gain (see Figure~\ref{fig:setting}). Similarly, companies hosting proprietary models (e.g., OpenAI, Google) could stealthily serve fine-tuned, weaker variants to reduce operational costs while evading detection. If these fine-tuning processes are computationally efficient and financially profitable, they pose a critical risk to API users. 
Therefore, we challenge the adversarial reliability of existing LLM fingerprinting methods and study a \textit{novel and fundamental research question}: 
\begin{graybox}
    \textit{Can a malicious model provider parameter-efficiently fine-tune a weak model so that the weak model behaves like a stronger model to bypass user-side fingerprinting tests?}
\end{graybox}

To answer this question, we formalize a threat model named \textbf{fingerprint spoofing}, where API service providers actively seek to bypass user-side model auditing (Figure~\ref{fig:main_fig}(a)). In this setting, the provider advertises a strong LLM to charge a premium rate, but actually serves a parameter-efficient~\cite{hu2022lora,dettmers2023qlora,liu2024dora} fine-tuned weak LLM to evade fingerprint checks. Successfully executing such an attack is highly non-trivial and presents severe technical challenges. Fundamentally, based on LLM scaling laws~\cite{kaplan2020scaling,hoffmann2022empirical,sardana2024beyond} and our lower-bound analysis (Proposition \ref{prop:spoofing_lb_informal}), a small model inherently lacks the representational capacity to perfectly mimic the behavior of a much larger model. Consequently, the attacker faces a severe optimization trade-off: aggressively fine-tuning a weak model to mimic a strong model's fingerprint behavior may lead to catastrophic forgetting~\cite{li2024revisiting,luo2025empirical}, destroying its ability to answer normal user queries. Furthermore, mimicking a large model with a small model must be achieved under resource constraints, since full-parameter retraining is often expensive and would make the adversary's attack financially infeasible. 

\begin{figure}[!t]
    \centering
    \includegraphics[width=0.95\linewidth]{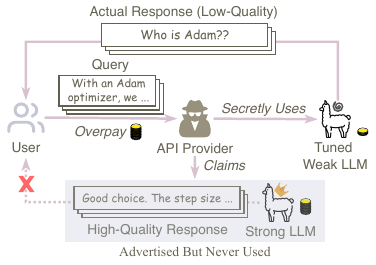}
    \vskip -1em
    \caption{\textbf{The economic motivation of LLM API spoofing}. To maximize profit, a malicious provider secretly serves a tuned weak LLM instead of the advertised strong model. If the provider successfully spoofs the user's fingerprint audit, the user is deceived into overpaying for lower-quality responses.}  
    \vskip -1em
    \label{fig:setting}
\end{figure}

Formalizing these technical challenges, our goal is to study whether a malicious provider can successfully resolve this trade-off to simultaneously achieve three conflicting objectives: (i) \textbf{fingerprint spoofing}, where the adapted model becomes behaviorally indistinguishable from the claimed strong model under auditing; (ii) \textbf{utility preservation}, where the adapted weak model maintains acceptable performance on normal downstream tasks; and (iii) \textbf{economic gain}, where the attack only relies on inexpensive parameter-efficient fine-tuning (PEFT) mechanisms.

To understand why spoofing remains feasible under the three goals, we theoretically analyze the vulnerabilities of current fingerprint pipelines. Constrained by user-side resources, users often rely on limited query sets~\cite{gubri2024trap,pasquini2025llmmap} and weak verification mechanisms~\cite{sun2025idiosyncrasies,gao2025model}. Consequently, an attacker will not need to emulate the strong model universally, and a localized, parameter-efficient adaptation suffices to fool the user-side audit. We formalize these limitations in Theorem \ref{thm:simple_dist_spoofing} (finite query budgets) and Theorem \ref{thm:classifier_spoofing} (weak local classifiers). Guided by these theoretical weaknesses, we propose a novel attack framework, namely {\bf GhostPrint}, to overcome the utility-spoofing trade-off. Specifically, the attacker trains a surrogate fingerprint model and employs reward-ranked fine-tuning to identify queries where the weak model can naturally mimic the larger model. We then leverage LoRA-based supervised fine-tuning (SFT) alongside knowledge distillation (KD) to align the weak model with the strong model, specifically on the targeted fingerprint queries, effectively bypassing the audit while preserving general utility. Furthermore, recognizing that user-side fingerprints evolve over time, we extend our framework to a continual spoofing setting, exploring whether an attacker can dynamically adapt to new fingerprint queries and auditing methods. We empirically validate this attack framework in both static and continual settings, exposing the practical risks posed by fingerprint spoofing.

Our \textbf{core contributions} are: (i) \textit{Theoretical Vulnerability Analysis}: We formally demonstrate when and why fingerprint spoofing is possible, proving that user-side resource constraints (i.e., weak query sets and weak surrogate models) make existing auditing methods fundamentally breakable (Theorem \ref{thm:simple_dist_spoofing} and \ref{thm:classifier_spoofing}); (ii) \textit{Fingerprint Spoofing Framework}: We propose a novel, practical attack that can 
successfully bypass static fingerprinting while maintaining task utility; and (iii) \textit{Continual Spoofing Evaluation}: We extend our attack to dynamic environments, demonstrating that malicious providers can continually adapt to evolving user-side fingerprinting methods, highlighting a persistent threat to LLM API integrity.
\section{Related Works}
\label{sec:related}

\noindent {\bf LLM Fingerprinting}.
Pay-per-use APIs create incentives for providers to silently substitute weaker models~\cite{cai2025you,zhang2026shadowapis}. Fingerprinting addresses this by verifying model identity. Existing methods fall into three categories: white-box non-intrusive approaches that inspect internal signals~\cite{zheng2022dnn,wu2025tensorguard,zhang2025reef,zhang2026attndiff}, black-box non-intrusive approaches that analyze API outputs~\cite{gao2025model,sun2025idiosyncrasies,pasquini2025llmmap,gubri2024trap}, and intrusive methods that embed verifiable signatures during training~\cite{gu2022watermarking,li2023plmmark,xu2024instructional,yamabe2025mergeprint}. Further discussion is provided in Appendix~\ref{sec:app:more_related}.

\noindent {\bf Attacks Against LLM Fingerprinting}.
Adversarial robustness of fingerprinting is an emerging concern. Prior attacks use rewriting or statistical evasion but assume no fine-tuning of the weak model and rely on distinguishing query distributions~\cite{kurian2025attacks,nasery2025robust}. In contrast, we consider a provider that actively fine-tunes a weak model to spoof fingerprints, a realistic threat that undermines distribution-agnostic auditors like LLM-idio~\cite{sun2025idiosyncrasies} and Model Equality Testing (MET)~\cite{gao2025model}.
\section{Preliminaries}\label{sec:prelim}

\noindent {\bf Notations}. 
We use $\|\Wb\|_F$ and $\|\Wb\|_2$ to denote the matrix Frobenius and $\ell_2$ norms, and use $\|\xb\|_2$ to denote the vector $\ell_2$ norm.
We treat an LLM as a conditional probability distribution over token sequences. Given an LLM $f$ parameterized by $\thetab$ and a prompt $q$, a response $r$ is sampled as $r \sim f(\cdot \mid q; \thetab)$. When the parameters are fixed, we write $r \sim f(\cdot \mid q)$ for notation simplicity. 

\noindent {\bf Black-Box LLM Fingerprinting}. The goal of black-box fingerprinting is to audit an LLM API, denoted as $f_{\mathrm{api}}$, to determine whether it actually serves a specific claimed target model, $f_{\mathrm{target}}$. The auditing procedure is formalized below.

\begin{definition}[Black-Box LLM Fingerprinting]\label{dfn:bb_llm_fingerprint}
    Let the user/auditor have a query distribution $\Dcal_{q}$, a query budget $n$, and a local classification mechanism $g$. The fingerprinting procedure is:
    \begin{enumerate}[label=\textbf{\arabic*.}, leftmargin=*, nosep, topsep=2pt]

        \item \textbf{Query Sampling}: The auditor samples a query set $\Qcal^n := \{q^i\}_{i=1}^n$, where $q^i \sim \Dcal_q$.

        \item \textbf{Response Collection}: The auditor queries the API and collects the responses $\Rcal^n := \{r^i\}_{i=1}^n$, where $r^i \sim f_{\mathrm{api}}(\cdot \mid q^i)$ and $q^i \in \mathcal{Q}^n$.

        \item \textbf{Verification}: The mechanism $g$ maps the collected responses to a binary decision $g(\Rcal^n, f_{\mathrm{target}}) \in \{0, 1\}$, returning $1$ if the API is verified as $f_{\mathrm{target}}$, and $0$ otherwise. 

    \end{enumerate}
\end{definition}

The above definition naturally encapsulates existing representative black-box fingerprinting methods, which differ mainly in their formulation of the query distribution $\Dcal_q$ and the verification mechanism $g$. For instance, LLMmap~\cite{pasquini2025llmmap} restricts $\Dcal_q$ to a small set of carefully crafted, identity-revealing probes and trains $g$ as a multi-class embedding classifier. 
MET~\cite{gao2025model} uses a general $\Dcal_q$ but formulates $g$ as a statistical two-sample kernel test (Maximum Mean Discrepancy) that computes the distance between $\Rcal^n$ and reference generations, thresholding the p-value to yield a binary decision. 

\noindent {\bf Practical Constraints of Fingerprinting}. In practice, the query budget $n$ in Definition~\ref{dfn:bb_llm_fingerprint} enforces a trade-off between auditing accuracy and financial cost, as users must pay for every API call. A larger $n$ is more accurate but result in higher cost.  Furthermore, user-side resource constraints typically restrict the verification mechanism $g$ to lightweight statistical tests~\cite{gao2025model,zhu2026auditing} or weak surrogate models~\cite{pasquini2025llmmap,sun2025idiosyncrasies} rather than state-of-the-art LLMs.

\noindent {\bf Why Black-box Fingerprinting?} 
We focus on black-box fingerprinting, where the auditor only has query access, because this mirrors the reality of commercial LLM APIs. End users cannot inspect internal model states or intervene in the provider's training pipeline. For proprietary models like GPT‑4o, the query‑response interface is the sole accessible surface, making black‑box methods the only practical tool for user‑side verification.
\section{Problem Formulation}
\label{sec:problem}


\subsection{Threat Model}\label{sec:threat_model}

We study a malicious third-party API provider (i.e., the adversary) that claims to serve a strong target model $f_{\mathrm{target}}$, but in reality deploys a weaker and cheaper model architecture $f_{\mathrm{weak}}$, initially parameterized by $\thetab_{\mathrm{weak}}$. The adversary aims to adapt these parameters into deceptive parameters $\thetab_{\mathrm{adv}}$ such that the auditor's local mechanism $g$ incorrectly verifies the adapted model's distribution $f_{\mathrm{weak}}(\cdot \mid \cdot; \thetab_{\mathrm{adv}})$ as $f_{\mathrm{target}}$.

\noindent {\bf Attacker's Capabilities.} We assume the attacker has the following capabilities:
\begin{itemize}[leftmargin=*, nosep, topsep=2pt]
    \item \textbf{Knowledge of the fingerprinting mechanism}: The attacker has full knowledge of the auditor's local classification mechanism $g$ (e.g., feature extractor, distance metric, etc.) and can build a surrogate $\ov{g}$ that mimics the auditor's $g$. 
    \item \textbf{Knowledge of the query distribution}: The attacker knows the distribution $\Dcal_q$ from which the auditor draws queries. This allows the attacker to sample a surrogate query set $\ov{\Qcal}^m := \{\ov{q}_i\}_{i=1}^m$, where $\ov{q}_i \sim \Dcal_q$. Importantly, \textit{the attacker does not know the exact finite query set $\Qcal^n$ instantiated by the auditor during the verification}. 
    
    \item \textbf{White-box manipulation of the weak model}: The attacker has full access to the deployed architecture $f_{\mathrm{weak}}$ and can tune its parameters via parameter-efficient fine-tuning (PEFT).
\end{itemize}

\noindent {\bf Attacker's Goal.} The attacker aims to simultaneously achieve: (i)  \textit{Fingerprint spoofing}: Under the fingerprinting mechanism (Definition~\ref{dfn:bb_llm_fingerprint}), the auditor's local mechanism $g$ incorrectly verifies $f_{\mathrm{weak}}(\cdot \mid \cdot; \thetab_{\mathrm{adv}})$ as $f_{\mathrm{target}}$; (ii) \textit{Utility preservation}: After adaptation, the served model $f_{\mathrm{weak}}(\cdot \mid \cdot; \thetab_{\mathrm{adv}})$ still maintains acceptable normal-task performance; and (iii) \textit{Economic gain}: The process to obtain $\thetab_{\mathrm{adv}}$ is cost-effective. 

\noindent {\bf Practicality of Threat Model.} 
Assuming the adversary knows both the query distribution $\Dcal_q$ and the classifier $g$ aligns with adaptive-attacker evaluations rooted in Kerckhoffs's principle. For an audit to be credible and independently verifiable, its methodology, classification algorithm $g$, and query set $\Dcal_q$ must be transparent and publicly disclosed. Relying on secrecy of the auditing pipeline constitutes security by obscurity, so a fully informed attacker model is necessary.

\subsection{Fingerprint Spoofing}\label{sec:spoofing_problem}

Based on the threat model described in Section~\ref{sec:threat_model}, we now formalize our attacker goals into a general problem definition. 

\begin{definition}[Fingerprint Spoofing Problem]\label{dfn:spoofing_problem}
    Given a white-box weak model architecture $f_{\mathrm{weak}}$ initially parameterized by $\thetab_{\mathrm{weak}}$, access to a target model $f_{\mathrm{target}}$, and knowledge of the auditor's query distribution $\Dcal_q$ and verification mechanism $g$, the fingerprint spoofing problem seeks to efficiently find adapted parameters $\thetab_{\mathrm{adv}}$ for the weak model to achieve three primary objectives:
    \begin{enumerate}[label=\textbf{\arabic*.}, leftmargin=*, nosep, topsep=2pt]
        \item \textbf{Fingerprint Spoofing:} The adapted model $f_{\mathrm{weak}}(\cdot \mid \cdot; \thetab_{\mathrm{adv}})$ must generate responses that fool the auditor's mechanism $g$ into identifying it as the target model over the distribution $\Dcal_q$.
        \item \textbf{Utility Preservation:} The adapted model $f_{\mathrm{weak}}(\cdot \mid \cdot; \thetab_{\mathrm{adv}})$ must maintain general model utility comparable to the initial weak model $f_{\mathrm{weak}}(\cdot \mid \cdot; \thetab_{\mathrm{weak}})$, ensuring the API remains functional for benign user requests.
        \item \textbf{Economic Gain:} The adversary's fine-tuning process to obtain $\thetab_{\mathrm{adv}}$ must be economically viable. Specifically, the adapted parameters are obtained via Parameter-Efficient Fine-Tuning (PEFT), i.e., $\thetab_{\mathrm{adv}} = \thetab_{\mathrm{weak}} + \Delta\thetab$, where $\Delta\thetab$ is a small low-rank update. 
    \end{enumerate}
\end{definition}

\section{Proposed Method}
\label{sec:proposed}

\begin{figure*}[t]
  \centering
  \includegraphics[width=0.99\linewidth]{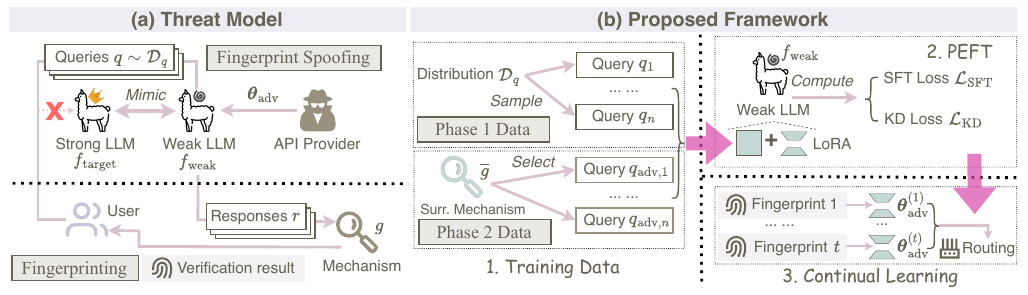}
  \vspace{-0.1in}
  \caption{Overview of GhostPrint. (a) Threat model. An adversary leverages a weak LLM to mimic a strong target model and evade fingerprint-based verification when interacting with an API provider. (b) Proposed framework. GhostPrint constructs training data from both target-model and surrogate queries, then fine-tunes a weak LLM via SFT and KD, and enables continual learning to progressively adapt the adversarial fingerprint over time.}
  \vspace{-0.1in}
  \label{fig:main_fig}
\end{figure*}

In this section, we first present a theoretical study on fundamental limits of LLM fingerprinting, and then present our proposed attack framework. An illustration of our framework is in Figure~\ref{fig:main_fig}. 

\subsection{When and Why Could Fingerprint Spoofing Happen?} \label{sec:theory}

As shown in Section~\ref{sec:threat_model}, fingerprint spoofing requires balancing three conflicting goals: fingerprint spoofing, utility preservation, and computational efficiency. To design a successful attack, it is critical to first establish a theoretical understanding of when and why fingerprint spoofing is feasible.

\noindent {\bf Impossibility of Universal Spoofing}. 
Fundamentally, small models lack the representational capacity of large models due to established scaling laws~\cite{kaplan2020scaling,hoffmann2022empirical,sardana2024beyond}. To formally capture this, we analyze a single high-rank weight transformation $\Wb$ (e.g., an FFN or attention projection matrix)~\cite{vaswani2017attention} in a target large model $f_{\mathrm{target}}$, and its corresponding lower-rank counterpart $\Wb'$ in the weak model $f_{\mathrm{weak}}$\footnote{In practice, weak models may use smaller hidden dimensions ($d_1' < d_1$) with weights $\Vb \in \mathbb{R}^{d_1' \times d_2}$. Our setting, where $\Wb$ and $\Wb'$ share identical dimensions but differ in rank, naturally captures this reality. By letting $\Wb' = \Pb\Vb$, where $\Pb \in \mathbb{R}^{d_1 \times d_1'}$ is an up-projection, the rank of $\Wb'$ is bottlenecked by $d_1'$, ensuring it has a lower rank than $\Wb$.}. The following proposition demonstrates that if an auditor has sufficient resources, spoofing via parameter-efficient fine-tuning (PEFT) is mathematically impossible. 

\begin{proposition}[Lower Bound of Universal Spoofing]
\label{prop:spoofing_lb_informal}
Let $\Wb, \Wb' \in \mathbb{R}^{d_1 \times d_2}$ be two weight matrices, where $\Wb$ has rank $k_1$ and $\Wb'$ has rank $k_2$, with $k_1 > k_2$. Let the low-rank adaptation (LoRA) matrices be $\Ab \in \mathbb{R}^{d_1\times r}$ and $\Bb \in \mathbb{R}^{r\times d_2}$, where $r < k_1 - k_2$. Define the adapted weight matrix as $\widehat{\Wb} := \Wb' + \Ab\Bb$ and the residual matrix as $\Rb := \Wb - \Wb'$. For any input vector $\xb \in \mathbb{R}^{d_2}$, the following lower bounds hold:
\begin{enumerate}[leftmargin=*, nosep, topsep=2pt]
    \item \textbf{Spectral Norm:}
    \begin{align*}
    \inf_{\Ab, \Bb} \sup_{\|\xb\|_2 = 1} \| \Wb\xb - \widehat{\Wb}\xb \|_2 \ge \sigma_{\min}^+(\Rb), 
    \end{align*}
    \item \textbf{Frobenius Norm:}
    \begin{equation*}\small
    \inf_{\Ab, \Bb} \| \Wb - \widehat{\Wb} \|_F^2 \ge (k_1 - k_2 - r) \cdot \sigma_{\min}^+(\Rb)^2,
    \end{equation*}
\end{enumerate}
where $\sigma_{\min}^+(\Rb)$ denotes the smallest non-zero singular value of $\Rb$.
\end{proposition}

The rank condition $k_2 < k_1$ reflects the scale gap between weak and strong models, and the bound on $r$ captures practical PEFT constraints. Proposition~\ref{prop:spoofing_lb_informal} shows universal spoofing error is strictly positive, motivating us to exploit specific weaknesses of real auditors. In practice, users have finite query budgets and weak classifiers, introducing vulnerabilities that make successful spoofing achievable despite the theoretical lower bound.

\noindent {\bf Vulnerabilities of LLM Fingerprinting}. Despite the theoretical impossibility of universal spoofing, users cannot expend infinite resources to audit an LLM API in practice. They cannot utilize infinite queries, nor can they always deploy the strongest classifiers to verify responses. These practical constraints create critical vulnerabilities and establish theoretical upper bounds for successful spoofing. 

\textbf{First}, existing auditing methods often rely on a small, simple set of queries. For example, LLMmap~\cite{pasquini2025llmmap} uses queries like ``Which model is this?'' with specific prefixes for auditing. When the input space is sufficiently simple and structured, we can establish the following upper bound based on query distribution simplicity. 

\begin{theorem}[Upper Bound under Low-Effective-Rank Distributions]
\label{thm:simple_dist_spoofing}
Let $\xb \in \mathbb{R}^{d_2}$ be the representation of a query drawn from a zero-mean continuous distribution $\Dcal_q$, and let $\mathbf{\Sigma} := \mathbb{E}_{\xb \sim \Dcal_q}[\xb\xb^\top]$ be the covariance matrix. Let $\lambda_1 \ge \lambda_2 \ge \dots \ge \lambda_{d_2} \ge 0$ denote the eigenvalues of $\mathbf{\Sigma}$. 
Let the LoRA matrices be $\Ab \in \mathbb{R}^{d_1\times r}$ and $\Bb \in \mathbb{R}^{r\times d_2}$, and the adapted weight matrix be $\widehat{\Wb} := \Wb' + \Ab\Bb$. The optimal expected spoofing error is bounded by the following equation: 
\vspace{-0.1in}
\begin{align*}\small
    \begin{aligned}
    \inf_{\Ab, \Bb} \mathbb{E}_{\xb \sim \Dcal_q} \big[ \| \Wb\xb - \widehat{\Wb}\xb \|_2^2 \big] \le \|\Rb\|_2^2 \sum_{i=r+1}^{d_2} \lambda_i, 
    \end{aligned}
\end{align*}
where $\Rb := \Wb - \Wb'$ is the residual matrix. 
\end{theorem} 

The theorem shows that when the query distribution $\Dcal_q$ is simple (small effective rank), the expected spoofing error can be substantially minimized. This matches LLMmap, where queries concentrate in specific semantic directions, revealing the first vulnerability of current auditing methods.

\textbf{Second}, due to training costs, auditors often rely on weak classifiers. For example,  LLMmap~\cite{pasquini2025llmmap} and LLM‑idio~\cite{sun2025idiosyncrasies} use small LLM‑based classifiers. To analyze this, we consider a projection weight matrix $\Cb$ from $g$'s transformer and test whether outputs $\Wb\xb$ and $(\Wb' + \Ab\Bb)\xb$ become indistinguishable under $\Cb$.

\begin{theorem}[Tight Bound under Weak Classifiers]
\label{thm:classifier_spoofing}
Let $\Cb \in \mathbb{R}^{d_1 \times d_1}$ be a weight matrix from the auditor's verification model, and let $\Wb, \Wb' \in \mathbb{R}^{d_1 \times d_2}$ be the weight matrices of the strong target model and the weak surrogate model, respectively. 
Let the LoRA matrices be $\Ab \in \mathbb{R}^{d_1\times r}$ and $\Bb \in \mathbb{R}^{r\times d_2}$, and the adapted weight matrix be $\wh{\Wb} := \Wb' + \Ab\Bb$. Under the audit, the worst-case spoofing error is exactly: 
\begin{align*}\small
    \begin{aligned}
    \inf_{\Ab, \Bb} \sup_{\|\xb\|_2 = 1} \| \Cb^\top (\Wb\xb - \wh{\Wb}\xb) \|_2 = \sigma_{r+1}(\Cb^\top \Rb),
    \end{aligned}
\end{align*}
where $\Rb = \Wb - \Wb'$ is the residual matrix, and $\sigma_{r+1}(\cdot)$ is the $(r+1)$-th largest singular value. 
\end{theorem}
If the auditor utilizes a simple classifier with rank $c \le r$, then $\sigma_{r+1}(\mathbf{C}^\top \Rb) = 0$. This shows that when the user-side fingerprint model $g$ is sufficiently weak, the adversary can ideally use PEFT to perfectly fool the fingerprint with zero error.

\noindent {\bf Takeaway.}
Our analysis reveals a critical misalignment: perfect spoofing is globally impossible (Proposition~\ref{prop:spoofing_lb_informal}), yet it is locally trivial under realistic auditing constraints (Theorems~\ref{thm:simple_dist_spoofing} and~\ref{thm:classifier_spoofing}). 

\subsection{The Proposed GhostPrint Attack}\label{sec:ghostprint}

As shown in Section~\ref{sec:theory}, practical LLM fingerprinting pipelines suffer from two fundamental vulnerabilities: (i) low-complexity query distributions $\Dcal_q$, and (ii) the reliance on computationally weak verification models $g$. A potent fingerprint spoofing attack must actively exploit both vulnerabilities simultaneously. To this end, we propose \textbf{GhostPrint}, a principled attack framework that leverages Supervised Fine-Tuning (SFT) and Knowledge Distillation (KD) to efficiently memorize the behavioral characteristics of the target model over $\Dcal_q$ while preserving normal utility. Furthermore, we incorporate a Ranked Fine-Tuning (RAFT)~\cite{dong2023raft} mechanism, utilizing a surrogate fingerprint model $\ov{g}$ to filter and select weak model outputs that match the strong model's fingerprint. To satisfy the attacker's economic constraints, all optimizations are conducted parameter-efficiently.

\noindent \textbf{Exploiting the Query Distribution $\Dcal_q$}. 
We employ Supervised Fine‑Tuning (SFT) to force the weak model to mimic the target model's fingerprinting behavior. For a fingerprint query $q \sim \Dcal_q$ and the target response $r$ generated by $f_{\mathrm{target}}$, the SFT objective minimizes the negative log‑likelihood:
\begin{align*}\small
    \begin{aligned}
    \mathcal{L}_{\mathrm{SFT}}(\thetab_{\mathrm{adv}}) := - \mathbb{E}_{q,r} \left[ {\sum}_{t=1}^{|r|} \log f_{\mathrm{weak}}(r_t \mid q, r_{<t}; \thetab_{\mathrm{adv}}) \right].
    \end{aligned}
\end{align*}
To preserve general capability without external data, we apply Knowledge Distillation (KD) on the same queries, using the target model as teacher:
\begin{align*}
    \begin{aligned}
    \mathcal{L}_{\mathrm{KD}}(\thetab_{\mathrm{adv}}) := \mathbb{E}_{q, r} \left[ {\sum}_{t=1}^{|r|} D_\mathrm{KL} \big( p_t^{\mathrm{target}} \big\|\; p_t^{\mathrm{weak}} \big) \right],
    \end{aligned}
\end{align*}
where $p_t^{\mathrm{target}}$ and $p_t^{\mathrm{weak}}$ are the next‑token distributions of the target and weak model, respectively. The total objective is
\begin{align}\label{eq:objective}
    \min_{\thetab_{\mathrm{adv}}} \mathcal{L}_{\mathrm{total}} := \mathcal{L}_{\mathrm{SFT}}(\thetab_{\mathrm{adv}}) + \alpha \mathcal{L}_{\mathrm{KD}}(\thetab_{\mathrm{adv}}),
\end{align}
where $\alpha > 0$ is a hyperparameter balancing fingerprint deception and utility.

\noindent \textbf{Exploiting the Fingerprint Mechanism $g$}. 
After optimizing~\eqref{eq:objective}, the weak model mimics the target on the query distribution $\Dcal_q$. However, when $\Dcal_q$ comprises diverse, general-purpose queries (e.g., UltraChat~\cite{ding2023ultrachat}), standard SFT alone may fail to capture the subtle model‑specific fingerprints required by auditors such as LLM‑idio~\cite{sun2025idiosyncrasies} and MET~\cite{gao2025model}. To address this, we directly exploit the structural weakness of the verification model $g$ by incorporating its supervisory signal into the fine‑tuning process.

Directly using $g$ is precluded under the black‑box threat model (Section~\ref{sec:threat_model}). Even if accessible, employing it as a reward model for reinforcement learning (RL)~\citep{guo2025deepseek,li2026ets,li2026reasoning,luo2026sparse} would be prohibitively expensive and unstable, violating the attacker’s economic constraints. We instead construct a local surrogate $\ov{g}$ from the known query distribution $\Dcal_q$, which can be a lightweight classifier (for LLMmap, LLM‑idio) or a reproduced statistical metric (for MET). To leverage $\ov{g}$ without RL, we adopt RAFT~\cite{dong2023raft} as an additional post‑training phase.

Specifically, for each query $q\sim\Dcal_q$, we sample $K$ candidate responses from the current weak model, rank them with $\ov{g}$, and select the most deceptive one as $r_{\mathrm{adv}} := \arg\max_{r_j} \ov{g}(r_j, f_{\mathrm{target}})$. The resulting $(q, r_{\mathrm{adv}})$ pairs are then used to fine‑tune the adapter parameters $\thetab_{\mathrm{adv}}$ with the same SFT and KD objectives defined in~\eqref{eq:objective}. This approach filters the weak model's natural generation space to isolate responses that already cross the surrogate's decision boundaries. By fine-tuning on these selected examples, the model leverages the signal from the target weak classifier $g$, while avoiding the unstable and costly use of RL.

\noindent {\bf Parameter-Efficient Fine-Tuning}. 
To address the economic gain goal shown in Section~\ref{sec:threat_model}, we restrict the attacker to apply Low-Rank Adaptation (LoRA) to the weak model rather than computationally prohibitive full-parameter fine-tuning. 

Let the initial parameters of the weak architecture $f_{\mathrm{weak}}$ be denoted as the set of $N$ targeted base weight matrices $\thetab_{\mathrm{weak}} := \{ \Wb'_i \}_{i=1}^N$. We formalize the deceptive parameter set $\thetab_{\mathrm{adv}}$ by assigning each adapted weight matrix $\widehat{\Wb}_i$ as the sum of the frozen base weight and a low-rank update:
\begin{align}\label{eq:adv_param}
    \thetab_{\mathrm{adv}} := \{ \widehat{\Wb}_i \mid \widehat{\Wb}_i = \Wb'_i + \Ab_i\Bb_i \}_{i=1}^N,
\end{align}
where $r < \min(d_{1,i}, d_{2,i})$ is the adapter rank, and $\Ab_i \in \mathbb{R}^{d_{1,i} \times r}$ and $\Bb_i \in \mathbb{R}^{r \times d_{2,i}}$ are the trainable LoRA matrices for the $i$-th target weight. 

By freezing $\thetab_{\mathrm{weak}}$ and optimizing the set $\Delta\thetab:=\{(\Ab_i, \Bb_i)\}_{i=1}^N$, the adversary transitions the generation distribution to $f_{\mathrm{weak}}(\cdot \mid \cdot; \thetab_{\mathrm{adv}})$ while reducing the computational cost of the attack.

\noindent {\bf Continual Learning of Multiple Fingerprints}. 
Our two-phase training bypasses a single fingerprint. However, auditors might continually deploy new methods. Thus, an adversary must spoof a sequence of verifiers $g^{(\tau)}$ with query distributions $\Dcal_q^{(\tau)}$ for $\tau=1,\dots,T$, updating $f_{\mathrm{weak}}$ without catastrophic forgetting. For each task at continual learning step $\tau$, we train a task‑specific LoRA adapter $\Delta\thetab^{(\tau)}:=\{(\Ab_i^{(\tau)}, \Bb_i^{(\tau)})\}_{i=1}^N$ via our aforementioned process, then integrate adapters with a Mixture‑of‑LoRA‑Experts (MoLA)~\cite{gao2025mola} architecture. It routes queries to relevant adapters, exploiting shared structure across fingerprinting mechanisms.

Specifically, for the $i$-th target weight matrix at continual learning step $\tau$, we introduce a trainable router matrix $\Vb_i^{(\tau)} \in \mathbb{R}^{t \times d_{2,i}}$, where $d_{2,i}$ is the dimension of the intermediate hidden state $\xb$. The router computes a $t$-dimensional logit vector to assign a routing score $s_{i,m}(\xb)$ for the $m$-th expert: 
\begin{align*}
    s_{i,m}(\xb) := \frac{\exp(\Vb_i^{(\tau)}\xb)_m}{\sum_{j =1}^t \exp(\Vb_i^{(\tau)}\xb)_j}.
\end{align*}

The adapted output $\hb_i^{(\tau)}$ for the $i$-th base weight matrix $\Wb'_i$ and input $\xb$ is then computed as a dynamically weighted combination of the historical LoRA experts as: 
\begin{align*}
    \hb_i^{(\tau)} := \Wb'_i \xb + \sum_{m = 1}^\tau s_{i,m}(\xb) \Ab_i^{(m)} \Bb_i^{(m)} \xb.
\end{align*} 

By freezing previous adapters $\Delta\thetab^{(<\tau)}:=\{\Delta\thetab^{(1)}, \Delta\thetab^{(2)}, \dots, \Delta\thetab^{(\tau-1)}\}$ and only training the new adapter $\Delta\thetab^{(\tau)}$ alongside the router $\Vb_i^{(t)}$, GhostPrint continually expands its spoofing capabilities across multiple detection methods while strictly maintaining parameter efficiency. 

To train the router matrices, we reuse the same SFT and KD objectives defined in Eq.~\eqref{eq:objective}. The training data for this routing phase is a mixture of queries from all tasks seen so far, i.e., $\Dcal_q^{(1)} \cup \Dcal_q^{(2)} \dots \cup \Dcal_q^{(\tau)}$. Because we only train the new adapter and a very small set of router weights $\Vb^{(\tau)} := \{\Vb_i^{(\tau)}\}_{i=1}^N$, this continual learning process remains highly efficient. 
\section{Experiments}
\label{sec:experiments}

\subsection{Experimental Settings}
\label{sec:single-fingerprint:settings}

\noindent\textbf{Models and Datasets.} We instantiate our attack across three representative LLM families selected from the model list of LLMmap~\cite{pasquini2025llmmap}. Specifically, we utilize Gemma-1.1 (2B as the weak model and 7B as the strong target)~\cite{gemma}, Qwen2 (1.5B as weak and 7B as strong)~\cite{qwen2}, and Phi-3 (Mini 3.8B as weak and Medium 14B as strong)~\cite{phi3}. To quantify the preservation of general capabilities, model utility is evaluated on three standard benchmarks under a 5-shot setting using the open-source LM Evaluation Harness\footnote{\url{https://github.com/EleutherAI/lm-evaluation-harness}}: MMLU~\cite{hendrycks2020measuring}, GSM8K~\cite{cobbe2021training}, and ARC-Challenge (ARC-C)~\cite{clark2018think}. 

\noindent\textbf{LLM Fingerprints.} We evaluate against three representative black-box LLM fingerprints to audit LLM APIs:
(i) {\it LLMmap}~\cite{pasquini2025llmmap}, a 52-class classifier driven by adversarial identity probes;
(ii) {\it LLM-idiosyncrasies (LLM-idio)}~\cite{sun2025idiosyncrasies}, a binary classifier trained on UltraChat~\cite{ding2023ultrachat} responses;
and (iii) {\it Model Equality Testing (MET)}~\cite{gao2025model}, a two-sample hypothesis test comparing completion distributions via MMD and Hamming distance.

\noindent\textbf{Spoofing Baselines.} We compare against four baselines: 
(i) {\it Query Detection}, where we train a simple MLP to memorize historical fingerprint queries, routing flagged queries to the strong model and all others to the weak model; 
(ii) {\it Rewriting}~\cite{kurian2025attacks}, which uses their proposed prompt to rewrite the generated response; 
(iii) {\it 5-shot ICL}, which extends the rewriting baseline by providing five strong-model responses as in-context examples; and 
(iv) {\it Token Suppression}~\cite{nasery2025robust}, which applies their Suppress Top-$K$ attack to randomize generation and mask the model's standard behavioral patterns. 
For reference, we also report the unmodified baseline performance of both the base weak model and the target strong model. 

\noindent\textbf{Evaluation Setting.} All fingerprint spoofing attacks are evaluated using an 8:2 train-test split for all fingerprint queries. For LLMmap~\cite{pasquini2025llmmap}, we use their official query set. For LLM-idio~\cite{sun2025idiosyncrasies} and MET~\cite{gao2025model}, we use the UltraChat~\cite{ding2023ultrachat} dataset as the fingerprint queries, consistent with the original papers. Due to space limitations, full implementation details are deferred to Appendix~\ref{sec:app:evaluation-details}.

\subsection{Single-Fingerprint Spoofing Performance}
\label{sec:single-fingerprint:main}

\begin{table}[t]
\small
\setlength{\tabcolsep}{4pt}
\renewcommand{\arraystretch}{1.15}
\resizebox{\linewidth}{!}{
\begin{tabular}{lcccccc}
\toprule
\multirow{2}{*}{Method} & \multicolumn{3}{c}{Utility} & \multicolumn{3}{c}{ASR (\%)} \\
\cmidrule(lr){2-4} \cmidrule(lr){5-7}
 & MMLU & GSM8K & ARC-C & LLMmap & LLM-idio & MET\\
\midrule
\multicolumn{7}{c}{\it Gemma-1.1-2B $\rightarrow$ Gemma-1.1-7B} \\
\midrule
Original (weak)          & 37.90 & 10.69 & 44.97 & 12.0          & 7.2           & 6.7  \\
Query Detection           & -- & -- & -- & 35.0          & 7.2           & 13.3 \\
Rewriting                 & -- & -- & -- & 20.0          & 0.7           & 0.0  \\
5-shot ICL                & -- & -- & -- & 15.0          & 22.9          & 3.3  \\
Token Suppression         & -- & -- & -- & 28.0          & 4.5           & 0.0  \\
GhostPrint (ours)         & 38.11 & 8.72  & 44.54 & \textbf{95.0}  & \textbf{58.5} & \textbf{23.3} \\
\cmidrule(lr){1-7}
Original (target) & 58.70 & 49.58 & 56.83 & 92.0          & 98.8          & 100.0 \\
\midrule
\multicolumn{7}{c}{\it Qwen2-1.5B $\rightarrow$ Qwen2-7B} \\
\midrule
Original (weak)          & 55.55 & 55.42 & 43.52 & 4.0           & 7.4           & 33.3 \\
Query Detection           & -- & -- & -- & 70.0          & 7.4           & \bf{40.0} \\
Rewriting                 & -- & -- & -- & 40.0          & 34.2          & 33.3 \\
5-shot ICL                & -- & -- & -- & 0.0           & 64.4          & 30.0 \\
Token Suppression         & -- & -- & -- & 6.0           & 4.9           & 6.7  \\
GhostPrint (ours)         & 55.36 & 55.80 & 44.97 & \textbf{76.0} & \textbf{71.3} & \bf{40.0} \\
\cmidrule(lr){1-7}
Original (target) & 70.67 & 72.93 & 62.12 & 81.0          & 98.6          & 96.7 \\
\midrule
\multicolumn{7}{c}{\it Phi-3-mini $\rightarrow$ Phi-3-medium} \\
\midrule
Original (weak)          & 70.37 & 80.36 & 62.97 & 6.0           & 28.8          & 43.3 \\
Query Detection           & -- & -- & -- & 70.0  & 28.8          & 50.0 \\
Rewriting                 & -- & -- & -- & 0.0           & 38.6          & 33.3 \\
5-shot ICL                & -- & -- & -- & 0.0           & \textbf{70.5} & 36.7 \\
Token Suppression         & -- & -- & -- & 5.0           & 51.2          & 6.7  \\
GhostPrint (ours)         & 69.70 & 78.01 & 62.80 & \textbf{80.0} & 64.0          & \textbf{60.0} \\
\cmidrule(lr){1-7}
Original (target) & 78.12 & 85.52 & 66.38 & 91.0          & 77.2          & 96.7 \\
\bottomrule
\end{tabular}}
\caption{
\textbf{Evaluation of single-fingerprint spoofing across three LLM families}. The table reports the Attack Success Rate (ASR) alongside general model utility. A successful attack maximizes ASR (best results bolded) while preserving utility scores comparable to the base weak model. Entries marked with ``--'' denote inference-time baselines that do not tune the weak model, where the model utility remains identical to the original weak model.
}
\label{tab:fingerprint-spoofing-main}
\vspace{-0.1in}
\end{table}
Table~\ref{tab:fingerprint-spoofing-main} reports the single-fingerprint spoofing performance of GhostPrint, with full implementation details of GhostPrint deferred to Appendix~\ref{sec:app:per-adapter}. The results demonstrate that our framework successfully achieves the attacker's primary goals shown in Section~\ref{sec:threat_model}. \textbf{(i)} Regarding \textit{fingerprint spoofing}, GhostPrint attains the highest Attack Success Rate (ASR) across all evaluated fingerprints for nearly all weak model families, with the only exception of the 5-shot ICL baseline on LLM-idio. Notably, when evaluating the Gemma family against LLMmap, GhostPrint achieves a 95\% ASR, even surpassing the original target model's natural classification rate of 92\%. \textbf{(ii)} In terms of \textit{utility preservation}, GhostPrint maintains downstream performance that is almost identical to the base weak models. While inference-time baselines inherently preserve utility by avoiding weight updates, our parameter-efficient fine-tuning incurs only marginal performance drops or even slight improvements, confirming that the attack successfully preserves the weak model's general capabilities on normal tasks.

\subsection{Cross-Model Spoofing Performance}
\label{sec:single-fingerprint:cross-model}

In this section, we extend the within model family setting in Section~\ref{sec:single-fingerprint:main} to the more challenging cross-family setting. Specifically, we study two transfer scenarios: Gemma-1.1-2B to Qwen2-7B, and Qwen2-1.5B to Gemma-1.1-7B. From the results in Table~\ref{tab:cross-model}, we observe: \textbf{(i)} Cross-family spoofing is inherently more difficult because different model architectures exhibit distinct baseline linguistic styles and vocabulary distributions; and \textbf{(ii)} Despite this architectural gap, GhostPrint still achieves notable success, successfully bypassing LLM-Idio in both settings and MET in the Gemma-to-Qwen transfer. 

For completeness, we also evaluate cross-family spoofing between models of the same scale in Appendix~\ref{sec:app:more-cross-model}. However, as this does not fit the financial gain goals for the adversary, it falls outside of our threat model and is put in the appendix.

\begin{table}[t]
\centering
\small
\setlength{\tabcolsep}{4pt}
\renewcommand{\arraystretch}{1.1}
\resizebox{\linewidth}{!}{
\begin{tabular}{llccc}
\toprule
\multirow{2}{*}{Setting} & \multirow{2}{*}{Method} & \multicolumn{3}{c}{ASR (\%)} \\
\cmidrule(lr){3-5}
 & & LLMmap & LLM-idio & MET \\
\midrule
\multirow{7}{*}{\shortstack[l]{Gemma-2B \\ $\to$ Qwen2-7B}}
 & Original (weak)   & 0.0   & 0.7   & 0.0   \\
 & Query Detection    & 0.0   & 2.4   & 6.7 \\
 & Rewriting          & 0.0   & 0.1   & 0.0 \\
 & 5-shot ICL         & 0.0   & 3.6   & 0.0 \\
 & Token Suppression  & 0.0   & 2.0   & 6.7 \\
 & GhostPrint (ours)  & 0.0   & \textbf{95.2} & \textbf{33.3} \\
\cmidrule(lr){2-5}
 & Original (target) & 80.0  & 99.9  & 96.7  \\
\midrule
\multirow{7}{*}{\shortstack[l]{Qwen2-1.5B \\ $\to$ Gemma-7B}}
 & Original (weak)   & 0.0   & 1.1   & 0.0   \\
 & Query Detection    & 0.0   & 2.5   & \textbf{6.7} \\
 & Rewriting          & 0.0   & 1.8   & 0.0 \\
 & 5-shot ICL         & 0.0   & 17.3  & 0.0 \\
 & Token Suppression  & \textbf{15.0} & 1.7   & 0.0 \\
 & GhostPrint (ours)  & 6.0   & \textbf{96.7} & \textbf{6.7}  \\
\cmidrule(lr){2-5}
 & Original (target) & 97.0  & 99.2  & 86.7  \\
\bottomrule
\end{tabular}}
\caption{\textbf{Evaluation of cross-family weak-to-strong spoofing}. The table reports the Attack Success Rate (ASR) when adapting a weak model from one architectural family to spoof a strong model from a completely different family. The highest ASR among attack methods is bolded. 
}

\vspace{-0.1in}
\label{tab:cross-model}
\end{table}

\subsection{Ablation Study}
\label{sec:single-fingerprint:kd}

\noindent\textbf{Knowledge-Distillation Loss.} 
In this study, we evaluate the stability of knowledge distillation under different coefficients $\alpha$ in Eq.~\eqref{eq:objective}. We also investigate whether the small model or the large model should serve as the reference anchor. Based on the results in Figure~\ref{fig:kd-ablation}, we find that: (i) using the large model as the anchor gives a slightly higher fingerprint spoofing ASR, but utilizing the small model remains highly effective. This demonstrates that our attack is effective in both scenarios. Whether the target model is accessed via responses only or through softmax probabilities that enable KD, the attack succeeds. (ii) The impact of the KD loss is generally stable across different weights, achieving the optimal ASR at approximately $\alpha = 0.25$.

\begin{figure}[!t]
  \centering
  \includegraphics[width=\linewidth]{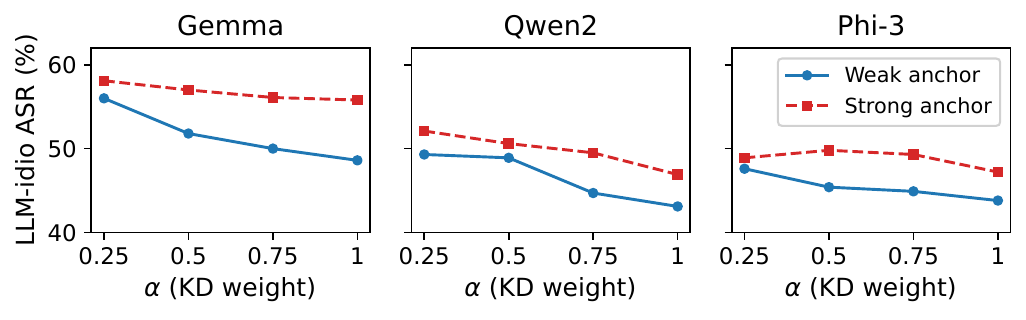}
  \caption{\textbf{Ablation study on the knowledge-distillation anchor model type and weight coefficient $\alpha$}. We report LLM-idio ASR (\%) for each family under the weak- and strong-model anchors.}
  \label{fig:kd-ablation}
\end{figure}

\label{sec:single-fingerprint:raft}

\noindent\textbf{Impact of RAFT }. 
To ensure that fingerprint spoofing fine-tuning effectively exploits the vulnerabilities of weak classifiers, we utilize RAFT to select queries that are more adversarial to fingerprint classifiers. The results are summarized in Table~\ref{tab:raft-ablation}. As shown in the table, incorporating a 50\% ratio of RAFT queries ($\lambda_{\mathrm{raft}}=0.5$) during fine-tuning alongside standard queries significantly improves the ASR. This demonstrates that strategically leveraging weak classifier vulnerabilities directly enhances spoofing capabilities, validating our theoretical analysis in Section~\ref{sec:theory}.

\begin{table}[t]
\centering
\small
\setlength{\tabcolsep}{4pt}
\renewcommand{\arraystretch}{1.15}
\resizebox{\linewidth}{!}{
\begin{tabular}{lcccc}
\toprule
\multirow{2}{*}{Family} & \multicolumn{2}{c}{LLM-idio ASR (\%)} & \multicolumn{2}{c}{MET ASR (\%)} \\
\cmidrule(lr){2-3} \cmidrule(lr){4-5}
& $\lambda_{\mathrm{raft}}=0$ & $0.5$ & $0$ & $0.5$ \\
\midrule
\textit{Gemma-1.1-2B$\to$7B}   & 54.7 & \textbf{58.8} & 6.7  & \textbf{10.0}     \\
\textit{Qwen2-1.5B$\to$7B}     & 56.3 & \textbf{63.4} & 36.7 & \textbf{40.0}   \\
\textit{Phi-3-mini$\to$medium} & 49.1 & \textbf{54.8} & 40.0 & \textbf{43.3} \\
\bottomrule
\end{tabular}}
\caption{\textbf{Ablation study on RAFT-guided query selection}. The parameter $\lambda_{\mathrm{raft}}$ represents the proportion of RAFT-selected queries mixed into the fine-tuning dataset, with the highest ASR bolded.}
\label{tab:raft-ablation}
\end{table}

\subsection{Continual Fingerprint Spoofing}
\label{sec:single-fingerprint:mola}

Here, we evaluate our mixture-of-LoRA-experts (MoLA) framework for the continual learning of multiple fingerprints. Specifically, we sequentially train three separate sets of LoRA parameters tailored to three distinct fingerprinting methods: {\it LLMmap}, {\it LLM-idio}, and {\it MET}. We then evaluate different combinations of these fingerprinted experts routed through a Mixture-of-Experts (MoE) architecture to see how effectively the combined parameters preserve and integrate capabilities. 

Our main results are summarized in Table~\ref{tab:mola}, from which we draw the following key observations: (i) After continual learning via integrating different parameters into the MoLA framework, performance experiences only minor degradation and remains at a high level. Notably, in the 3-way configurations for Qwen and Phi-3, some fingerprint spoofing results even outperform their standalone counterparts. This demonstrates that our framework can continually adapt to new fingerprints while retaining prior knowledge without catastrophic forgetting. (ii) Even when an expert is not explicitly trained on a specific fingerprint, the integrated MoLA framework demonstrates an inherent ability to generalize. For instance, in the {LLMmap} + {LLM-idio} routing configuration on Qwen2 and Phi-3, despite the complete omission of the {MET} expert during sequential training, the evaluation scores on {MET} closely approach those of configurations where it was explicitly included. This suggests that the compositionality of the learned experts facilitates cross-fingerprint knowledge transfer.

\begin{table}[t]
\centering
\small
\setlength{\tabcolsep}{4pt}
\renewcommand{\arraystretch}{1.1}
\resizebox{\linewidth}{!}{
\begin{tabular}{llccc}
\toprule
\multirow{2}{*}{Family} & \multirow{2}{*}{Router Config} & \multicolumn{3}{c}{Attack ASR (\%)} \\
\cmidrule(lr){3-5}
 & & LLMmap & LLM-idio & MET \\
\midrule
\multirow{4}{*}{\textit{Gemma-1.1-2B$\to$7B}}
 & LLMmap + LLM-idio    & 86.0 & 54.2 & 10.0  \\
 & LLMmap + MET         & \textbf{89.0} & 13.2 & 10.0  \\
 & LLM-idio + MET       & 20.0 & \textbf{58.6} & \textbf{30.0} \\
 & 3-way (all)          & 86.0 & 54.8 & 16.7 \\
\midrule
\multirow{4}{*}{\textit{Qwen2-1.5B$\to$7B}}
 & LLMmap + LLM-idio    & 75.0 & \textbf{58.4} & 23.3  \\
 & LLMmap + MET         & 74.0 & 55.5 & 16.7  \\
 & LLM-idio + MET       & 20.0 & 57.1 & 26.7  \\
 & 3-way (all)          & \textbf{81.0} & 57.6 & \textbf{26.7}  \\
\midrule
\multirow{4}{*}{\textit{Phi-3-mini$\to$medium}}
 & LLMmap + LLM-idio    & 71.0 & \textbf{53.0} & 43.3 \\
 & LLMmap + MET         & 70.0 & 27.5          & 43.3          \\
 & LLM-idio + MET       & 0.0           & 49.3          & \textbf{60.0}          \\
 & 3-way (all)          & \textbf{77.0} & 50.6 & 53.3 \\
\bottomrule
\end{tabular}}
\caption{\textbf{Continual fingerprint spoofing performance using MoLA across different task settings}. Bold values indicate the highest ASR achieved within each model family and column.}
\label{tab:mola}
\vspace{-0.1in}
\end{table}

\section{Conclusion}

We introduce fingerprint spoofing and propose {GhostPrint}, a parameter-efficient attack combining surrogate modeling, reward-ranked fine-tuning, and knowledge distillation, extended with a Mixture-of-LoRA router to unify multiple auditors.
GhostPrint consistently bypasses four auditors across three model families while preserving utility, exposing a critical vulnerability in current black-box LLM fingerprint methods.

\section*{Limitations}
Due to academic compute limits, GhostPrint is not evaluated on frontier-scale models (e.g., Qwen-3 32B), which is an important next step.
Future work may incorporate reinforcement learning with the surrogate classifier as a reward model to improve alignment.
While our attack exposes a vulnerability that could be misused, this is a research prototype intended to motivate robust fingerprinting and we do not foresee a potential negative societal impact.
LLMs were used only for grammar checking and language quality, and the authors are fully responsible for all technical content.

\bibliography{ref}


\clearpage
\appendix

\begin{center}
    {\LARGE \bf Appendix}
\end{center}

\section*{List of Contents}
In this appendix, we provide the following additional information: 

\begin{itemize}[leftmargin=0.05\linewidth,nosep]
    \applist
\end{itemize}

\appsection{Implementation Details of Experiments}
\label{sec:app:evaluation-details}

\subsection{Fingerprinting Auditors}
\label{sec:app:fingerprints}
For LLMmap, we use the released set of 800 adversarial identity probes, each concatenated with a system prompt and CoT/RAG‑style prompts as in the original work. The 800 probes are split 8:2 (640 training, 160 test), and we use 100 probes to evaluate ASR.

For LLM-idio, we follow the original setting: a per‑family binary classifier is trained on 10,000 UltraChat~\cite{ding2023ultrachat} responses from each of the small and large models. The same prompts are used for training the attack adapters and for evaluation; a held‑out 20\% subset (2,000 prompts) serves as the test set.

For MET, the two‑sample test is applied to completions from the UltraChat prompt set. For the main‑table fair‑budget setting we use $N=1$ prompt and $K=10$ completions per prompt, subsampled from the full $N=50, K=20$ collection. The test statistic is $\mathrm{MMD}_{\mathrm{hamming}}$ with a permutation $p$‑value (100 permutations, $\alpha=0.05$). ASR is the fraction of 30 independent random seeds for which MET fails to reject $H_0$. The full $(N,K)$ sweep is described in Appendix~\ref{sec:app:met-complexity}.

\subsection{Baseline Settings}
\label{sec:app:baselines}

For Query Detection, we train a defender-side classifier that decides at audit time whether to serve the small or the large model. The classifier is a two-layer MLP on top of frozen \texttt{all-MiniLM-L6-v2} sentence embeddings, trained to distinguish LLMmap identity probes from UltraChat prompts. At evaluation, each query is routed to the large model's pre-collected response if flagged as a probe and to the small model's response otherwise. This is an oracle baseline because the defender knows the probe distribution; utility is unchanged on prompts the detector rejects.

For Rewriting~\citep{kurian2025attacks}, the small model first answers the query, then is re-prompted with a fixed style-laundering template that asks it to rewrite the response ``to be more comprehensive, detailed, and well-reasoned, as if written by a highly knowledgeable assistant with deep expertise.'' Both passes use the same small model; only the second pass output is audited. Sample sizes per auditor match the Original (small) row.

For 5-shot ICL, we sample five question-answer pairs from the large model's UltraChat responses and prepend them as an in-context demonstration block before every query that the small model answers. For LLMmap, the same five pairs are prepended to every identity probe before evaluation; for LLM-idio and MET, they are prepended to every UltraChat prompt.

For Token Suppression~\citep{nasery2025robust}, we follow the SuppressTop-$k$ attack. At decoding time, for the first $n=16$ generation positions we mask the $k=4$ highest-probability tokens (setting their logits to $-\infty$) before sampling; subsequent positions decode normally. The mask is implemented as a HuggingFace \texttt{LogitsProcessor} and applied uniformly across all three auditors, using the same temperature, prompt set, and sample count as Original (small).

\subsection{Training Settings}
We adopt the LoRA-based deception attack from Section~\ref{sec:proposed}, with $r=16$, $\alpha=32$, dropout $0.05$, $2$ epochs, learning rate $5\times 10^{-5}$, batch size $4$, AdamW with weight decay $0.01$, and a knowledge-distillation term ($\lambda_{\mathrm{kd}}=0.5$, $\tau=2.0$) against a frozen large teacher. We set $\lambda_{\mathrm{kd}}=0.5$ as the midpoint of the sweep in Figure~\ref{fig:kd-ablation}, which shows that ASR is insensitive to $\lambda_{\mathrm{kd}}$ across $[0.25, 1.0]$. The RAFT-style ablation of Section~\ref{sec:single-fingerprint:raft} uses $N=8$ candidates per prompt, $3$ rounds, and $\lambda_{\mathrm{raft}}=0.5$.

\subsection{MET Sample Complexity}\label{sec:app:met-complexity}

LLMmap and LLM-idio use one query per audit; MET defaults to $(N=20,K=5)$, costing $100\times$ more generations. The main paper therefore reports MET at the fair budget $(1,10)$. Here we characterize detection power by subsampling from a $(50,20)$ pool, computing ASR over $30$ seeds.

\begin{figure}[t]
  \centering
  \includegraphics[width=\linewidth]{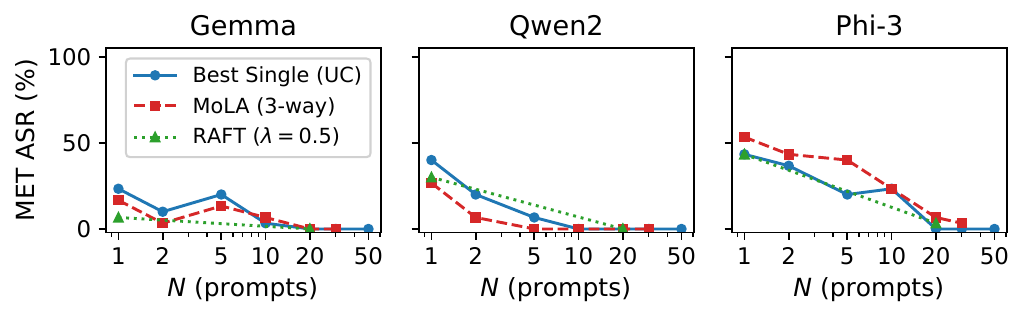}
  \caption{\textbf{MET sample complexity at $K=10$ for three spoofing attacks across model families}. ASR is averaged over 30 resampling seeds.}
  \label{fig:met-complexity}
\end{figure}

At the strong budget on Phi‑3, our MoLA adapter achieves a mean permutation $p=0.086>\alpha$, the first time MET fails to reject $H_0$. Its ASR stays above $50\%$ until $N>30$, while the best single‑fingerprint adapter drops below $50\%$ at $N=20$. On other families, curves overlap, consistent with family‑dependent trade‑offs.

Even though all adapters are rejected at the full $(50,20)$ budget, that setting costs $1{,}000$ generations per audit and is impractical for frequent use. Auditors operating at lower budgets, including the fair $(1,10)$ and the intermediate $(20,5)$ default, leave many attacks undetected, especially on Phi‑3 where the scale gap is smallest.

\appsection{Additional Experimental Results}\label{app:results}
\subsection{Per-adapter Single-Fingerprint Breakdown}\label{sec:app:per-adapter}

Table~\ref{tab:per-adapter} reports each of the 18 single-fingerprint LoRA variants we train per family, organised by training data (LLMmap probes vs.\ UltraChat) and KD teacher (no KD, student-size self-distillation, target-size teacher). The ``GhostPrint (ours)'' rows in Table~\ref{tab:fingerprint-spoofing-main} pick the column-wise maximum per family across these 18 cells.

\begin{table}[!ht]
\centering
\small
\setlength{\tabcolsep}{4pt}
\renewcommand{\arraystretch}{1.1}
\resizebox{\linewidth}{!}{
\begin{tabular}{lllccc}
\toprule
Family & Training Data & KD Teacher & LLMmap & LLM-idio & MET \\
\midrule
\multirow{6}{*}{Gemma}
& LLMmap     & --        & \textbf{95.0}  & 31.2          & 13.3          \\
& LLMmap     & 2B        & 84.0           & 14.7          & 3.3           \\
& LLMmap     & 7B        & 87.0           & 16.8          & 0.0           \\
& UltraChat  & --        & 19.0           & \textbf{58.5} & \textbf{23.3} \\
& UltraChat  & 2B        & 21.0           & 52.0          & 13.3          \\
& UltraChat  & 7B        & 28.0           & 55.9          & 13.3          \\
\midrule
\multirow{6}{*}{Qwen2}
& LLMmap     & --        & 74.0           & 70.0          & 26.7          \\
& LLMmap     & 1.5B      & \textbf{76.0}  & 41.5          & 30.0          \\
& LLMmap     & 7B        & \textbf{76.0}  & \textbf{71.3} & 36.7          \\
& UltraChat  & --        & 21.0           & 50.2          & \textbf{40.0} \\
& UltraChat  & 1.5B      & 28.0           & 51.0          & 36.7          \\
& UltraChat  & 7B        & 35.0           & 51.7          & \textbf{40.0} \\
\midrule
\multirow{6}{*}{Phi-3}
& LLMmap     & --        & \textbf{80.0}  & \textbf{64.0} & 26.7          \\
& LLMmap     & mini      & 77.0           & 35.9          & 43.3          \\
& LLMmap     & medium    & 74.0           & 35.2          & 33.3          \\
& UltraChat  & --        & 10.0           & 58.0          & 33.3          \\
& UltraChat  & mini      & 5.0            & 53.0          & 36.7          \\
& UltraChat  & medium    & 0.0            & 56.0          & \textbf{60.0} \\
\bottomrule
\end{tabular}}
\caption{\textbf{Detailed evaluation across all 18 LoRA variants}. The results are categorized by model family and training configurations (dataset $\times$ KD teacher selection). The best performance within each family and column is bolded.}
\label{tab:per-adapter}
\end{table}

\subsection{More Cross-Model Spoofing Results}\label{sec:app:more-cross-model}

\begin{table}[t]
\centering
\small
\setlength{\tabcolsep}{4pt}
\renewcommand{\arraystretch}{1.1}
\resizebox{\linewidth}{!}{
\begin{tabular}{llccc}
\toprule
\multirow{2}{*}{Setting} & \multirow{2}{*}{Method} & \multicolumn{3}{c}{ASR (\%)} \\
\cmidrule(lr){3-5}
 & & LLMmap & LLM-idio & MET \\
\midrule
\multirow{7}{*}{\shortstack[l]{Gemma-2B \\ $\to$ Qwen2-1.5B}}
 & Original (weak)   & 0.0   & 1.6   & 0.0   \\
 & Query Detection    & 0.0   & 3.2   & 6.7 \\
 & Rewriting          & 0.0   & 0.0   & 3.3 \\
 & 5-shot ICL         & 0.0   & 0.8   & 0.0 \\
 & Token Suppression  & 1.0   & 2.2   & 16.7 \\
 & GhostPrint (ours)  & \textbf{4.0} & \textbf{97.5} & \textbf{46.7} \\
\cmidrule(lr){2-5}
 & Original (target) & 78.0  & 99.5  & 100.0 \\
\midrule
\multirow{7}{*}{\shortstack[l]{Qwen2-1.5B \\ $\to$ Gemma-2B}}
 & Original (weak)   & 0.0   & 0.8   & 0.0   \\
 & Query Detection    & 0.0   & 2.5   & 6.7 \\
 & Rewriting          & 0.0   & 8.9   & 0.0 \\
 & 5-shot ICL         & 0.0   & 19.4  & 0.0 \\
 & Token Suppression  & 16.0  & 1.0   & 3.3 \\
 & GhostPrint (ours)  & \textbf{38.0} & \textbf{97.9} & \textbf{20.0} \\
\cmidrule(lr){2-5}
 & Original (target) & 84.0  & 99.1  & 93.3  \\
\bottomrule
\end{tabular}}
\caption{\textbf{Evaluation of cross-family weak-to-weak spoofing}. The table reports the Attack Success Rate (ASR) when adapting a weak model from one architectural family to spoof a strong model from a completely different family. The highest ASR among attack methods is bolded. 
}

\vspace{-0.1in}
\label{tab:cross-model-more}
\end{table}

For completeness, we provide additional results to investigate whether models of similar scale can serve as fingerprints for other model families. As shown in Table~\ref{tab:cross-model-more}, our proposed GhostPrint framework consistently outperforms all baselines. Notably, our results on LLM-idio across both settings are particularly notablez, showing almost no performance difference compared to the target model's original fingerprint.
\appsection{Additional Related Works}\label{sec:app:more_related}

\subsection{LLM API Overcharging Risks}
The proliferation of commercial LLM APIs under pay-per-use pricing has raised concerns about provider dishonesty at multiple levels. 
At the model level, providers may silently substitute promised strong models with weaker alternatives~\cite{cai2025you,zhang2026shadowapis}, directly misrepresenting service quality.
At the token level, hidden reasoning tokens in opaque APIs introduce invisible billing overhead~\cite{sun2025coin,sun2025invisible}, with tokenization transparency and incentive misalignment further analyzed in~\cite{velasco2025your}.
Broader resource-side threats, such as compute, memory, and API abuse, are surveyed in~\cite{zhang2026resource}.
Early defenses address both directions: a game-theoretic auditing mechanism targets dishonest model-level providers~\cite{cao2026pay}, while statistical and predictive auditing schemes tackle pay-per-token billing fraud~\cite{velasco2026auditing,wang2025predictive}.
Our work identifies a new threat distinct from previous model substitution and token overcharging, where a malicious provider can use parameter-efficient fine-tuning to adapt a weak model to evade fingerprint-based auditing, making dishonest model deployments harder to detect.

\subsection{LLM Fingerprinting} 

\noindent {\bf White-box Non-Intrusive Fingerprinting}. 
White-box fingerprinting assumes direct weight access and verifies ownership through internal model analysis without modifying model behavior.
Parameter and gradient signatures form one primary axis: \citet{zheng2022dnn} establish non-repudiable ownership certificates from weight statistics, and TensorGuard~\cite{wu2025tensorguard} classifies model families via gradient-based features.
Representation-based approaches extract fingerprints from internal activations: REEF~\cite{zhang2025reef} encodes layer-wise hidden states as compact fingerprints, AttnDiff~\cite{zhang2026attndiff} exploits differential attention maps for identification, and EasyDetector~\cite{zhang2024easydetector} detects provenance via linear probes on intermediate features.
HuRef~\cite{zeng2024huref} constructs human-interpretable fingerprints from weight-invariant vectors, and \citet{shao2025sok} provides a systematic survey of LLM copyright auditing spanning the full white-box landscape.

\noindent {\bf Non-Intrusive Black-Box Fingerprinting}. 
Black-box fingerprinting operates through API queries alone, without parameter access.
\textit{General-query methods} derive model identity from naturally elicited responses.
For API integrity auditing, statistical hypothesis tests detect model substitution by comparing output distributions~\cite{gao2025model,zhu2026auditing} and behavioral consistency~\cite{cai2025you,nikolic2025model}.
Idiosyncrasy-based approaches exploit persistent behavioral quirks~\cite{sun2025idiosyncrasies,mcgovern2025your}, characteristic error-space signatures~\cite{zang2025errortrace}, and side-channel query statistics~\cite{pasquini2025llmmap} as inherent fingerprints.
\textit{Specific-query methods} craft targeted inputs to elicit discriminative signals.
Adversarial honeypot strategies embed verifiable traps for IP attribution~\cite{gubri2024trap,jin2024proflingo}.
Token-level probing exploits under-trained tokens~\cite{cai2025utf} and training corpus artifacts~\cite{zhang2025speculating} as covert signals.
Further directions include membership-inference-based provenance~\cite{kuditipudi2025blackbox}, chain-of-thought fingerprinting~\cite{ren2025cotsrf}, zeroth-order gradient estimation~\cite{shao2026reading}, sensitivity-guided query selection~\cite{hu2025resf}, code model tracing~\cite{zhang2025clmtracing}, and dual-level IP protection frameworks~\cite{yan2026duffin}.

\noindent {\bf Intrusive Fingerprinting}. 
Intrusive methods embed verifiable fingerprints during model training via backdoor mechanisms.
Early work establishes backdoor watermarking baselines for pre-trained language models~\cite{gu2022watermarking,li2023plmmark}.
Instruction-triggered backdoors bind ownership responses to specific prompt patterns~\cite{xu2024instructional,xu2025ctcc}, and LoRA-based transfer propagates fingerprints into fine-tuned derivatives~\cite{xu2025unlocking}.
Robustness under post-hoc model modification is a central challenge: merge-resistant designs~\cite{yamabe2025mergeprint,he2025routemark} and fine-tuning-resistant schemes~\cite{tang2025towards,nasery2025scalable} address editing and adaptation threats.
Semantically richer fingerprints are pursued via conditioned watermarks~\cite{gloaguen2026llm}, random-seed fingerprinting~\cite{tong2026seedprints}, and functional evolution tracing~\cite{wu2026llm}.
Additional contributions include encrypted ownership verification~\cite{xiong2026iseal}, plug-and-play fingerprint generation~\cite{chen2026prompt2fingerprint}, open-source misuse detection~\cite{xu2025mark}, tamper attribution~\cite{bai2025esf}, and analysis of the fundamental difficulty of black-box origin identification~\cite{yang2025challenge}.

\newcommand{\rank}{\mathrm{rank}}
\appsection{Detailed Proofs}\label{sec:app:detail_proof}

{\bf Notations.} For a set of vectors $\{\mathbf{x}_1, \dots, \mathbf{x}_k\}$, $\mathrm{span}\{\mathbf{x}_1, \dots, \mathbf{x}_k\}$ denotes the set of all their linear combinations. For a matrix $\mathbf{A} \in \mathbb{R}^{d_1 \times d_2}$, we use $\mathrm{col}(\mathbf{A}) := \{\mathbf{A}\mathbf{x} \mid \mathbf{x} \in \mathbb{R}^{d_2}\}$ to denote its column space, and $\mathrm{null}(\mathbf{A}) := \{\mathbf{x} \in \mathbb{R}^{d_2} \mid \mathbf{A}\mathbf{x} = \mathbf{0}\}$ to denote its null space. Finally, $\tr(\mathbf{A})$ denotes the trace of a square matrix $\mathbf{A}$, defined as the sum of its diagonal elements.

\subsection{Basic Tools}

\begin{lemma}[The Eckhart-Young Theorem, adapted from page 79 of~\cite{golub2013matrix}]\label{lem:eckhart_young}
Let $\mathbf{A} \in \mathbb{R}^{d_1 \times d_2}$ be a matrix with rank $r$, and let its singular value decomposition be given by $\mathbf{A} = \sum_{i=1}^r \sigma_i \mathbf{u}_i \mathbf{v}_i^\top$, where $\sigma_1 \ge \sigma_2 \ge \dots \ge \sigma_r > 0$ are the non-zero singular values, and $\mathbf{u}_i, \mathbf{v}_i$ are the corresponding left and right singular vectors. 

For any integer $k < r$, let $\mathbf{A}_k:= \sum_{i=1}^k \sigma_i \mathbf{u}_i \mathbf{v}_i^\top$ denote the rank-$k$ truncated singular value decomposition of $\mathbf{A}$. 

Then, the optimal rank-$k$ approximation of $\mathbf{A}$ under the Frobenius norm and spectral norm is $\mathbf{A}_k$, i.e., 
\begin{align*}\small
    \begin{aligned}
    \argmin_{\mathrm{rank}(\mathbf{B})=k} \|\mathbf{A} - \mathbf{B}\|_2 = \argmin_{\mathrm{rank}(\mathbf{B})=k} \|\mathbf{A} - \mathbf{B}\|_F = \Ab_k. 
    \end{aligned}
\end{align*}
\end{lemma}

By substituting matrix $\Bb$ with the optimal solution $\Ab_k$, we obtain the following corollary with basic algebra. 

\begin{corollary}[Error of Optimal Low-Rank Approximation]\label{cor:eckhart_young_err}
    Under the same setting as Lemma~\ref{lem:eckhart_young}, the minimum approximation errors achieved by a rank-$k$ matrix are:
    \begin{enumerate}[leftmargin=*, nosep, topsep=2pt]
    \item \textbf{Spectral Norm:}
    \begin{align*}
    \min_{\mathrm{rank}(\mathbf{B})=k} \|\mathbf{A} - \mathbf{B}\|_2 = \sigma_{k+1},
    \end{align*}
    \item \textbf{Frobenius Norm:}
    \begin{align*}
    \min_{\mathrm{rank}(\mathbf{B})=k} \|\mathbf{A} - \mathbf{B}\|_F = \left( \sum_{i=k+1}^r \sigma_i^2 \right)^{1/2}.
    \end{align*}
    \end{enumerate}
\end{corollary}

\subsection{Proof of Proposition~\ref{prop:spoofing_lb_informal}}

\begin{proposition}[Lower Bound of Universal Spoofing, Restatement of Proposition~\ref{prop:spoofing_lb_informal}]

Let $\Wb, \Wb' \in \mathbb{R}^{d_1 \times d_2}$ be two weight matrices, where $\Wb$ has rank $k_1$ and $\Wb'$ has rank $k_2$, with $k_1 > k_2$. Let the low-rank adaptation (LoRA) matrices be $\Ab \in \mathbb{R}^{d_1\times r}$ and $\Bb \in \mathbb{R}^{r\times d_2}$, where $r < k_1 - k_2$. Define the adapted weight matrix as $\widehat{\Wb} := \Wb' + \Ab\Bb$ and the residual matrix as $\Rb := \Wb - \Wb'$. For any input vector $\xb \in \mathbb{R}^{d_2}$, the following lower bounds hold:
\begin{enumerate}[leftmargin=*, nosep, topsep=2pt]
    \item \textbf{Spectral Norm:} \begin{equation*}
    \inf_{\Ab, \Bb} \sup_{\|\xb\|_2 = 1} \| \Wb\xb - \widehat{\Wb}\xb \|_2 \ge \sigma_{\min}^+(\Rb), 
    \end{equation*}
    \item \textbf{Frobenius Norm:}
    \begin{equation*}
    \inf_{\Ab, \Bb} \| \Wb - \widehat{\Wb} \|_F^2 \ge (k_1 - k_2 - r) \cdot \sigma_{\min}^+(\Rb)^2,
    \end{equation*}
\end{enumerate}
where $\sigma_{\min}^+(\Rb)$ denotes the smallest non-zero singular value of $\Rb$.
\end{proposition}
\begin{proof}
{\bf Part 1. Spectral Norm}. We first consider the LHS of the spectral norm bound without taking the infimum: 
\begin{align*}
    \sup_{\|\xb\|_2 = 1} \|\Wb\xb - \widehat{\Wb}\xb\|_2 = & ~ \|\Wb - \widehat{\Wb}\|_2\\
    = & ~\|\Rb - \Ab\Bb\|_2,
\end{align*}
where the first equality follows from the definition of spectral norm, and the second equality follows from the definition of $\Rb$. 

Taking the infimum over the LoRA matrices $\Ab$ and $\Bb$ on both sides, we obtain:
\begin{align}\label{eq:lb_part1_eq1}
    \inf_{\Ab, \Bb} \sup_{\|\xb\|_2 = 1} \|\Wb\xb - \widehat{\Wb}\xb\|_2 = \inf_{\Ab, \Bb} \|\Rb - \Ab\Bb\|_2.
\end{align}

Since the LoRA matrices $\Ab \in \mathbb{R}^{d_1\times r}$ and $\Bb \in \mathbb{R}^{r\times d_2}$ are bottlenecked by their inner dimension $r$, their product $\Ab\Bb$ has a rank of at most $r$. Therefore, minimizing $\|\Rb - \Ab\Bb\|_2$ is equivalent to finding the optimal rank-$r$ approximation of the residual matrix $\Rb$ under the spectral norm. We can directly apply part 1 of Corollary~\ref{cor:eckhart_young_err} to evaluate this optimal error, which yields:
\begin{align}\label{eq:lb_part1_eq2}
    \inf_{\Ab, \Bb} \|\Rb - \Ab\Bb\|_2 = \sigma_{r+1}(\Rb).
\end{align}

We now lower bound $\sigma_{r+1}(\Rb)$. Because $\Wb = \Wb' + \Rb$, where $\mathrm{rank}(\Wb) = k_1$ and $\mathrm{rank}(\Wb') = k_2$, the subadditivity property of matrix rank dictates that $\mathrm{rank}(\Rb) \ge k_1 - k_2$. The proposition assumes the adaptation rank satisfies $r < k_1 - k_2$, which directly implies $r + 1 \le \mathrm{rank}(\Rb)$. Consequently, the $(r+1)$-th singular value must be at least as large as the smallest non-zero singular value of $\Rb$:
\begin{align} \label{eq:lb_part1_eq3}
    \sigma_{r+1}(\Rb) \ge \sigma_{\min}^+(\Rb).
\end{align}

Combining the Eq.~\eqref{eq:lb_part1_eq1}, Eq.~\eqref{eq:lb_part1_eq2} and Eq.~\eqref{eq:lb_part1_eq3} above, we obtain the desired result:
\begin{align*}
    \inf_{\Ab, \Bb} \sup_{\|\xb\|_2 = 1} \| \Wb\xb - \widehat{\Wb}\xb \|_2 \ge \sigma_{\min}^+(\Rb),
\end{align*}
which finishes the proof for the spectral norm case.

{\bf Part 2. Frobenius Norm}. We consider the LHS of the Frobenius norm bound and have the following:  
\begin{align*}
\inf_{\Ab, \Bb} \| \Wb - \widehat{\Wb} \|_F^2 = & ~ \inf_{\Ab,\Bb} \|\Rb - \Ab\Bb\|_F^2\\
= & ~ \min_{\rank(\Xb) \le r} \|\Rb - \Xb\|_F^2 \\
= & ~ \sum_{i=r+1}^{\rank(\Rb)} \sigma_i^2(\Rb) \\
\ge & ~ (\rank(\Rb) - r) \sigma_{\min}^+(\Rb)^2 \\
\ge & ~ (k_1 - k_2 - r) \sigma_{\min}^+(\Rb)^2, 
\end{align*}
where the first step is from the definition of $\Rb$, the second step is from the fact that $\Ab\Bb$ is of rank $r$, the third step is from part 2 of Corollary~\ref{cor:eckhart_young_err}, the fourth step is from the fact that all terms in the summation are at least $\sigma_{\min}^+(\Rb)^2$, and the last step is from the fact on $\rank(\Rb) \ge k_1 - k_2$. 

Thus, we complete the proof.
\end{proof}

\subsection{Proof of Theorem~\ref{thm:simple_dist_spoofing}}

To prove this upper bound, we first present a useful fact. 
\begin{fact}[Folklore] \label{fac:E_AxxA}
For a zero-mean stochastic vector $\mathbf{x}$ with covariance matrix $\mathbf{M} := \mathbb{E}[\mathbf{x}\mathbf{x}^\top]$ and any constant matrix $\mathbf{A}$, we have:
\begin{align*}
\mathbb{E}[(\mathbf{A}\mathbf{x})(\mathbf{A}\mathbf{x})^\top] = \mathbf{A}\mathbf{M}\mathbf{A}^\top.
\end{align*}
\end{fact}

Now, we show the upper bound under the simple query distribution setting. 
\begin{theorem}[Upper Bound under Low-Effective-Rank Distributions, Restatement of Theorem~\ref{thm:simple_dist_spoofing}]
Let $\xb \in \mathbb{R}^{d_2}$ be the representation of a query drawn from a zero-mean continuous distribution $\Dcal_q$, and let $\mathbf{\Sigma} := \mathbb{E}_{\xb \sim \Dcal_q}[\xb\xb^\top]$ be the covariance matrix. Let $\lambda_1 \ge \lambda_2 \ge \dots \ge \lambda_{d_2} \ge 0$ denote the eigenvalues of $\mathbf{\Sigma}$. 

Let the LoRA matrices be $\Ab \in \mathbb{R}^{d_1\times r}$ and $\Bb \in \mathbb{R}^{r\times d_2}$, and the adapted weight matrix be $\widehat{\Wb} := \Wb' + \Ab\Bb$. The optimal expected spoofing error is bounded by the following: 
\begin{align*}
    \begin{aligned}
    \inf_{\Ab, \Bb} \mathbb{E}_{\xb \sim \Dcal_q} \big[ \| \Wb\xb - \widehat{\Wb}\xb \|_2^2 \big] \le \|\Rb\|_2^2 \sum_{i=r+1}^{d_2} \lambda_i, 
    \end{aligned}
\end{align*}
where $\Rb := \Wb - \Wb'$ is the residual matrix. 
\end{theorem} 
\begin{proof}
Let $\mathbf{\Sigma} = \mathbb{E}_{\xb\sim\Dcal_q}[\xb\xb^\top]$ have eigendecomposition $\mathbf{\Sigma} = \Ub \Lb \Ub^\top$ with $\Lb = \diag(\lambda_1,\dots,\lambda_{d_2})$ and $\lambda_1 \ge \lambda_2 \ge \dots \ge 0$.  Choose $\Ub_r \in \R^{d_2 \times r}$ as the first $r$ eigenvectors.  Construct
\begin{align*}
\Ab\Bb = \Rb \Ub_r \Ub_r^\top,
\end{align*}
which has rank at most $r$ and thus can be factored as $\Ab \Bb$.  

For this choice, we have the following:
\begin{align*}
    & ~ \mathbb{E}_{\xb \sim \Dcal_q} \big[ \| \Wb\xb - \widehat{\Wb}\xb \|_2^2 \big] \\
    = & ~ \mathbb{E}_{\xb \sim \Dcal_q} \big[ \| (\Rb - \Ab\Bb) \xb \|_2^2 \big] \\ 
    = & ~ \mathbb{E}_{\xb \sim \Dcal_q} \big[ \| \Rb (\Ib - \Ub_r \Ub_r^\top) \xb \|_2^2 \big] \\ 
    = & ~ \mathbb{E}_{\xb \sim \Dcal_q} \big[ \| \Rb (\Ub\Ub^\top - \Ub_r \Ub_r^\top) \xb \|_2^2 \big] \\ 
    = & ~ \mathbb{E}_{\xb \sim \Dcal_q} \left[ \left\| \Rb \left(\sum_{j=r+1}^{d_2}\ub_j \ub_j^\top\right) \xb  \right\|_2^2 \right],
\end{align*}
where the first step is from the definition of $\widehat{\Wb}$ and $\Rb$, the second step is from choosing a specific $\Ab\Bb$, the third step is from the orthogonality of eigenvectors, and the last step is from expanding the matrix into eigenvector outer products and eliminating the first $r$ terms. 

For simplicity, let $\Qb := \sum_{j=r+1}^{d_2}\ub_j \ub_j^\top$. Applying the trace identity for expected quadratic forms in Fact~\ref{fac:E_AxxA}, we have:
\begin{align*}
    \mathbb{E}_{\xb \sim \Dcal_q} \big[ \| \Wb\xb - \widehat{\Wb}\xb \|_2^2 \big]  = & ~ \mathbb{E}_{\xb \sim \Dcal_q} \big[ \| \Rb\Qb\xb \|_2^2 \big]\\
    = & ~ \tr(\Rb\Qb\mathbf{\Sigma}\Qb^\top\Rb^\top).
\end{align*}

Since eigenvectors form an orthonormal basis, we expand and collapse the inner product $\Qb\mathbf{\Sigma}\Qb^\top$ as follows:
\begin{align*}\small
    \begin{aligned}
    & ~ \Qb\mathbf{\Sigma}\Qb^\top \\
    = & ~ \left(\sum_{j=r+1}^{d_2} \ub_j \ub_j^\top\right) \left(\sum_{i=1}^{d_2} \lambda_i \ub_i \ub_i^\top\right) \left(\sum_{k=r+1}^{d_2} \ub_k \ub_k^\top\right) \\
    = & ~ \sum_{i=r+1}^{d_2} \lambda_i \ub_i \ub_i^\top,
    \end{aligned}
\end{align*}
where the second step follows from the orthogonality condition $\ub_j^\top \ub_i = 1$ if $i=j$ and $0$ otherwise.

Finally, we substitute this collapsed term back into the trace to bound the expected error:
\begin{align*}
    & ~ \mathbb{E}_{\xb \sim \Dcal_q} \big[ \| \Wb\xb - \widehat{\Wb}\xb \|_2^2 \big] \\
     = & ~ \tr(\Rb\Qb\mathbf{\Sigma}\Qb^\top\Rb^\top)\\
     = & ~ \tr\left(\Rb \left(\sum_{i=r+1}^{d_2} \lambda_i \ub_i \ub_i^\top\right) \Rb^\top\right) \\
    = & ~ \sum_{i=r+1}^{d_2} \lambda_i \tr(\Rb \ub_i \ub_i^\top \Rb^\top) \\
    = & ~ \sum_{i=r+1}^{d_2} \lambda_i \tr(\ub_i^\top \Rb^\top \Rb \ub_i) \\
    = & ~ \sum_{i=r+1}^{d_2} \lambda_i \|\Rb\ub_i\|_2^2 \\
    \le & ~ \sum_{i=r+1}^{d_2} \lambda_i \|\Rb\|_2^2 \|\ub_i\|_2^2 \\
    = & ~ \|\Rb\|_2^2 \sum_{i=r+1}^{d_2} \lambda_i,
\end{align*}
where the third step follows from the linearity of the trace operator, the fourth step follows from the cyclic property of the trace, the fifth step follows from the definition of the $\ell_2$ norm, the fifth step follows from the fact that $\|\Rb\ub_i\|_2 \le \|\Rb\|_2\|\ub_i\|_2$, and the final step follows from the fact that eigenvectors are unit vectors (i.e., $\|\ub_i\|_2 = 1$).

Taking the infimum over all parameterizations of $\Ab$ and $\Bb$ can only yield an expected error less than or equal to this specific construction. Therefore, the upper bound holds, completing the proof.
\end{proof}

\subsection{Proof of Theorem~\ref{thm:classifier_spoofing}}

To show the tight bound under the weak classifier setting, we first prove a variant of the Eckhart-Young Theorem. 

\begin{lemma}[A Variant of the Eckhart-Young Theorem]\label{lem:eckhart_young_variant}
Let $\mathbf{A} \in \mathbb{R}^{d_1 \times d_2}$ and $\Cb \in \mathbb{R}^{d_1\times d_1}$. Let the singular value decomposition of the product matrix $\Cb\mathbf{A}$ be given by $\Cb\mathbf{A} = \sum_{i=1}^p \sigma_i \mathbf{u}_i \mathbf{v}_i^\top$, where $p = \min(d_1, d_2)$ and $\sigma_1 \ge \sigma_2 \ge \dots \ge 0$. 

For any integer $k < p$, let $(\Cb\mathbf{A})_k:= \sum_{i=1}^k \sigma_i \mathbf{u}_i \mathbf{v}_i^\top$ denote the rank-$k$ truncated singular value decomposition of $\Cb\mathbf{A}$. 

Then, the minimum achievable approximation error under the spectral norm is:
\begin{align*}
    \min_{\mathrm{rank}(\mathbf{B}) \le k} \|\Cb\mathbf{A} - \Cb\mathbf{B}\|_2  = \sigma_{k+1}(\Cb\mathbf{A}).
\end{align*}
\end{lemma}
\begin{proof}

{\bf Part 1. Existence}. We show that there exists a rank-$k$ matrix $\mathbf{B}$ such that $\|\Cb\mathbf{A} - \Cb\mathbf{B}\|_2 = \sigma_{k+1}(\Cb\mathbf{A})$. 

Because the product matrix $\Cb\mathbf{A}$ is formed by linear combinations of the columns of $\Cb$, its column space is restricted: $\mathrm{col}(\Cb\mathbf{A}) \subseteq \mathrm{col}(\Cb)$. Furthermore, its truncated SVD is built from its principal directions, meaning $\mathrm{col}((\Cb\mathbf{A})_k) \subseteq \mathrm{col}(\Cb\mathbf{A})$. 

Since $\mathrm{col}((\Cb\mathbf{A})_k) \subseteq \mathrm{col}(\Cb)$, the linear equation $\Cb\mathbf{X} = (\Cb\mathbf{A})_k$ is consistent. We can construct a specific solution using the Moore-Penrose pseudoinverse: let $\mathbf{B}^* = \Cb^+ (\Cb\mathbf{A})_k$. This satisfies $\Cb\mathbf{B}^* = (\Cb\mathbf{A})_k$. Furthermore, because the rank of a product cannot exceed the rank of its factors, $\mathrm{rank}(\mathbf{B}^*) \le \mathrm{rank}((\Cb\mathbf{A})_k) \le k$.

We now consider this specific valid rank-$k$ matrix $\mathbf{B}^*$ and compute the approximation error:
\begin{align*}
\|\Cb\mathbf{A} - \Cb\mathbf{B}^*\|_2 = & ~ \|\Cb\mathbf{A} - (\Cb\mathbf{A})_k\|_2 \\
= & ~ \left\| \sum_{i=1}^p \sigma_i \mathbf{u}_i \mathbf{v}_i^\top - \sum_{i=1}^k \sigma_i \mathbf{u}_i \mathbf{v}_i^\top \right\|_2 \\
= & ~ \left\| \sum_{i=k+1}^p \sigma_i \mathbf{u}_i \mathbf{v}_i^\top \right\|_2,
\end{align*}
which follows from basic algebra. 

Because the remaining singular values are ordered $\sigma_{k+1} \ge \sigma_{k+2} \ge \dots$, the spectral norm of this residual matrix is exactly the largest remaining singular value. Thus, $\|\Cb\mathbf{A} - \Cb\mathbf{B}^*\|_2 = \sigma_{k+1}(\Cb\mathbf{A})$.

{\bf Part 2. Optimality}. We show that no matrix $\mathbf{B}$ with rank at most $k$ can achieve a lower error value than $\sigma_{k+1}(\Cb\mathbf{A})$. 

Suppose $\mathrm{rank}(\mathbf{B}) \le k$ for some arbitrary matrix $\mathbf{B} \in \mathbb{R}^{d_1 \times d_2}$. Let $\mathbf{Y} = \Cb\mathbf{B}$. By the properties of matrix multiplication, $\mathrm{rank}(\mathbf{Y}) \le \mathrm{rank}(\mathbf{B}) \le k$. 

Because $\mathbf{Y}$ has $d_2$ columns and rank at most $k$, its null space has a dimension of at least $d_2 - k$, which follows from the rank-nullity theorem. It follows that we can find orthonormal vectors $\mathbf{x}_1, \dots, \mathbf{x}_{d_2-k}$ such that $\mathrm{null}(\mathbf{Y}) = \mathrm{span}\{\mathbf{x}_1, \dots, \mathbf{x}_{d_2-k}\}$ via the Gram-Schmidt process. 

Let $\mathcal{V}_{k+1} := \mathrm{span}\{\mathbf{v}_1, \dots, \mathbf{v}_{k+1}\}$ be the $(k+1)$-dimensional subspace spanned by the top right singular vectors of $\Cb\mathbf{A}$. A dimension argument shows that:
\begin{align*}
(d_2 - k) + (k + 1) = & ~ d_2 + 1 > d_2,
\end{align*}
meaning the sum of their dimensions exceeds the ambient space $\mathbb{R}^{d_2}$. Therefore, their intersection must be non-trivial:
\begin{align*}
\mathrm{span}\{\mathbf{x}_1, \dots, \mathbf{x}_{d_2-k}\} \cap \mathrm{span}\{\mathbf{v}_1, \dots, \mathbf{v}_{k+1}\} \neq \{\mathbf{0}\}.
\end{align*}

Let $\mathbf{z}$ be a unit 2-norm vector (i.e., $\|\mathbf{z}\|_2 = 1$) in this intersection. Since $\mathbf{z} \in \mathrm{null}(\mathbf{Y})$, we have $\mathbf{Y}\mathbf{z} = \Cb\mathbf{B}\mathbf{z} = \mathbf{0}$. Furthermore, since $\mathbf{z} \in \mathcal{V}_{k+1}$, applying $\Cb\mathbf{A}$ to $\mathbf{z}$ yields:
\begin{align}\label{eq:cbz}
(\Cb\mathbf{A})\mathbf{z} = & ~ \sum_{i=1}^{k+1} \sigma_i (\mathbf{v}_i^\top \mathbf{z}) \mathbf{u}_i.
\end{align}

We can now lower-bound the spectral norm for any choice of $\mathbf{B}$:
\begin{align*}
\|\Cb\mathbf{A} - \Cb\mathbf{B}\|_2^2 \ge & ~ \|(\Cb\mathbf{A} - \Cb\mathbf{B})\mathbf{z}\|_2^2 \\
= & ~ \|(\Cb\mathbf{A})\mathbf{z} - \mathbf{0}\|_2^2 \\
= & ~ \left\| \sum_{i=1}^{k+1} \sigma_i (\mathbf{v}_i^\top \mathbf{z}) \mathbf{u}_i \right\|_2^2 \\
= & ~ \sum_{i=1}^{k+1} \sigma_i^2 (\mathbf{v}_i^\top \mathbf{z})^2 \\
\ge & ~ \sigma_{k+1}^2 \sum_{i=1}^{k+1} (\mathbf{v}_i^\top \mathbf{z})^2 \\
= & ~ \sigma_{k+1}^2 \|\mathbf{z}\|_2^2 \\
= & ~ \sigma_{k+1}^2,
\end{align*}
where the sixth step follows from the fact that $\zb \in \mathrm{span}\{\vb_1, \dots, \vb_{k+1}\}$, and other steps follows from definitions and basic algebra. 

Taking the square root of both sides yields $\|\Cb\mathbf{A} - \Cb\mathbf{B}\|_2 \ge \sigma_{k+1}(\Cb\mathbf{A})$, completing the proof.

\end{proof}

Based on the lemma above, we show the tight bound as follows. 

\begin{theorem}[Tight Bound under Weak Classifiers, Restatement of Theorem~\ref{thm:classifier_spoofing}]
Let $\Cb \in \mathbb{R}^{d_1 \times d_1}$ be a weight matrix from the auditor's verification model, and let $\Wb, \Wb' \in \mathbb{R}^{d_1 \times d_2}$ be the weight matrices of the strong target model and the weak surrogate model, respectively. 

Let the LoRA matrices be $\Ab \in \mathbb{R}^{d_1\times r}$ and $\Bb \in \mathbb{R}^{r\times d_2}$, and the adapted weight matrix be $\wh{\Wb} := \Wb' + \Ab\Bb$. Under the audit, the worst-case spoofing error is exactly: 
\begin{align*}
    \begin{aligned}
    \inf_{\Ab, \Bb} \sup_{\|\xb\|_2 = 1} \| \Cb^\top (\Wb\xb - \wh{\Wb}\xb) \|_2 = \sigma_{r+1}(\Cb^\top \Rb),
    \end{aligned}
\end{align*}
where $\Rb = \Wb - \Wb'$ is the residual matrix, and $\sigma_{r+1}(\cdot)$ is the $(r+1)$-th largest singular value. 
\end{theorem}
\begin{proof}
By the definition of $\wh{\Wb}$ and $\Rb$, and the definition of the matrix spectral norm, we first rewrite the worst-case spoofing error:
\begin{align} \label{eq:tight_eq1}
& ~ \inf_{\Ab, \Bb} \sup_{\|\xb\|_2 = 1} \| \Cb^\top (\Wb\xb - \wh{\Wb}\xb) \|_2 \notag \\
= & ~ \inf_{\Ab, \Bb} \|\Cb^\top\Rb - \Cb^\top\Ab\Bb\|_2.
\end{align}

Because the LoRA matrices $\Ab \in \mathbb{R}^{d_1\times r}$ and $\Bb \in \mathbb{R}^{r\times d_2}$ are bottlenecked by their inner dimension $r$, their product $\Ab\Bb$ forms a matrix of rank at most $r$. Conversely, any $d_1 \times d_2$ matrix with a rank of at most $r$ can be exactly factored into the product of two such matrices $\Ab$ and $\Bb$. Therefore, taking the infimum over the parameter space of $\Ab$ and $\Bb$ is mathematically equivalent to minimizing over all possible matrices $\Xb$ of rank at most $r$:
\begin{align}\label{eq:tight_eq2}
& ~ \inf_{\Ab, \Bb} \|\Cb^\top\Rb - \Cb^\top\Ab\Bb\|_2 \notag \\
= & ~ \min_{\mathrm{rank}(\Xb) \le r} \|\Cb^\top\Rb - \Cb^\top\Xb\|_2.
\end{align}

We now apply Lemma~\ref{lem:eckhart_young_variant}, substituting $\Cb$ with $\Cb^\top$, $\mathbf{A}$ with $\Rb$, and the rank constraint $k$ with $r$. This yields exactly:
\begin{align} \label{eq:tight_eq3}
\min_{\mathrm{rank}(\Xb) \le r} \|\Cb^\top\Rb - \Cb^\top\Xb\|_2 = & ~ \sigma_{r+1}(\Cb^\top\Rb).
\end{align}

Combining Eq.~\eqref{eq:tight_eq1}, Eq.~\eqref{eq:tight_eq2}, and Eq.~\eqref{eq:tight_eq3} completes the proof.
\end{proof}

\end{document}